\documentclass[pdflatex,12pt]{article}

\usepackage{authblk}
\usepackage{graphicx}
\usepackage{latexsym}
\usepackage{amsmath,amssymb} 
\usepackage{algorithm}
\usepackage{algpseudocode}
\usepackage{tikz}

\usepackage{color}

\algrenewcommand\algorithmicrequire{\textbf{Input:}}
\algrenewcommand\algorithmicensure{\textbf{Output:}}

\newcommand{\itemmap}{\sigma}
\newcommand{\Item}{I_\sigma}
\newcommand{\Max}{{\mathcal C}_{\mathrm{max}}}

\newcommand{\MS}{{\mathcal S}}
\newcommand{\MC}{{\mathcal C}}
\newcommand{\MB}{{\mathcal B}}

\renewcommand{\algref}[1]{Algorithm~\ref{alg:#1}}
\newcommand{\angleb}[1]{\langle#1\rangle}
\newcommand{\corref}[1]{Corollary~\ref{cor:#1}}
\newcommand{\figref}[1]{Figure~\ref{fig:#1}}

\newcommand{\lemref}[1]{Lemma~\ref{lem:#1}}
\newcommand{\lineref}[1]{line~\ref{line:#1}}
\newcommand{\thmref}[1]{Theorem~\ref{thm:#1}}
\newcommand{\secref}[1]{Section~\ref{sec:#1}}

\newcommand{\Ind}{\mathsf{Ind}}

\newtheorem{cor}{Corollary}
\newtheorem{lem}{Lemma}
\newtheorem{thm}{Theorem}

\newenvironment{proof}{
\noindent{\scshape Proof:}}{\quad $\Box$\bigskip}

\long\def\invis#1{}

\title{Design of Polynomial-delay Enumeration Algorithms in Transitive Systems}
\author[1]{Kazuya Haraguchi}
\author[2]{Hiroshi Nagamochi}
\date{}
\affil[1]{Otaru University of Commerce}
\affil[ ]{haraguchi@res.otaru-uc.ac.jp}
\affil[2]{Department of Applied Mathematics and Physics, Kyoto University}
\affil[ ]{nag@amp.i.kyoto-u.ac.jp} 

\begin{document} 
\maketitle

\begin{abstract}
In this paper, as a new notion, we define a {\em transitive system}
to be a set system $(V, \MC\subseteq 2^V)$ on a finite set $V$
of elements such that 
 every three sets $X,Y,Z\in\MC$ with 
$Z\subseteq X\cap Y$ implies $X\cup Y\in\MC$,
where we call a set $C\in \MC$  a {\em component}. 
We assume that two oracles $\mathrm{L}_1$ and $\mathrm{L}_2$
are available, where  given
two subsets $X,Y\subseteq  V$, $\mathrm{L}_1$ returns 
a maximal component $C\in \MC$
with $X\subseteq C\subseteq Y$;
and 
given a set $Y\subseteq  V$, $\mathrm{L}_2$ returns 
all maximal components $C\in \MC$ with $C\subseteq Y$. 
Given a set $I$ of attributes and a function $\sigma:V\to 2^I$
in a transitive system,
a component $C\in \MC$ is called a {\em solution}
if the set of common attributes in $C$ is inclusively maximal; 
 i.e., $\bigcap_{v\in C}\sigma(v)\supsetneq \bigcap_{v\in X}\sigma(v)$
 for any component $X\in\MC$ with $C\subsetneq X$.
We prove that there exists an algorithm of enumerating all solutions
in delay bounded by a polynomial with respect to
the input size and the running times of the oracles.
The proposed algorithm yields
the first polynomial-delay algorithms
for enumerating connectors in an attributed graph 
and for enumerating all subgraphs
with various types of connectivities such as 
 all $k$-edge/vertex-connected  induced subgraphs
 and  
 all $k$-edge/vertex-connected  spanning subgraphs
in a given undirected/directed graph 
for a fixed $k$.
\end{abstract}

\section{Introduction}
\label{sec:intro}

In the present paper, we introduce a novel notion of set system,
``a transitive system.''
For a transitive system on a set of elements
  and  a set of items (or attributes)  given to each element,
we design an algorithm that enumerates all subsets in the system
that are inclusion-wise maximal with respect to the common items in a subset. 

Let $V$ be  a finite set  of elements.
A {\em system} on a set $V$ of elements is defined to be
a pair  $(V,\MC)$ of   $V$ of elements and   
a family $\MC\subseteq 2^V$,
where  
a set in $\MC$ is called a {\em component}.
For a subset $X\subseteq V$ in a system $(V,\MC)$, 
a component $Z\in \MC$ with $Z\subseteq X$ is called {\em  $X$-maximal} 
if  no other component $W\in \MC$ satisfies $Z\subsetneq W\subseteq X$,
and let $\Max(X)$  denote the family of all $X$-maximal components. 
For two subsets $X\subseteq  Y\subseteq V$,
let $\Max(X;Y)$ denote the family of components $C\in \Max(Y)$
such that $X\subseteq C$.
We call a system $(V,\MC)$ (or $\MC$)   {\em transitive} if  
\[\mbox{any tuple of  components  $X,Y,Z\in \MC$ with $Z\subseteq X\cap Y$
 implies $X\cup Y\in\MC$.}\]
For example, any Sperner family,
a family of subsets every two of which intersect, is
a transitive system. 
We call a set function  $\rho$ from $2^V$ to the set $\mathbb{R}$  of reals
a {\em volume function} if 
$\rho(X)\leq \rho(Y)$ for any subsets $X\subseteq Y\subseteq V$.
A subset $X\subseteq V$ is called {\em $\rho$-positive} if $\rho(X)>0$.
To discuss the computational complexities for solving a problem in a  transitive system,
we assume that a transitive system $(V,\MC)$ is implicitly given as two oracles 
$\mathrm{L}_1$ and  $\mathrm{L}_2$ such that
\begin{itemize}
\item[-]
  given non-empty  subsets $X\subseteq Y\subseteq V$, 
  $\mathrm{L}_1(X,Y)$  returns a component $Z\in \Max(X;Y)$ 
  (or $\emptyset$ if no such $Z$ exists)
  in $\theta_{\mathrm{1,t}}$ time and  
  $\theta_{\mathrm{1,s}}$ space; and 
\item[-] 
  given a non-empty subset $Y\subseteq V$, 
  $\mathrm{L}_2(Y)$  returns $\Max(Y)$
  in $\theta_{\mathrm{2,t}}$ time and  
  $\theta_{\mathrm{2,s}}$ space.
\end{itemize}
Given a volume function $\rho$,
we assume that whether $\rho(X)> 0$ holds or not can be
tested in $\theta_{\rho,\mathrm{t}}$ time 
and $\theta_{\rho,\mathrm{s}}$ space. 
We also denote by $\delta(X)$ an upper bound on $|\Max(X)|$,
where we assume that $\delta$ is a non-decreasing function in the sense that
$\delta(Y)\leq \delta(X)$ holds for any subsets $Y\subseteq X\subseteq V$.

We define an {\em instance} to be a tuple $\mathcal{I}=(V,\MC,I,\sigma)$  of  
a set $V$ of $n\geq 1$ elements, a family $\MC\subseteq 2^V$, 
a set $I$ of $q\geq 1$ items and a function $\sigma:V\to 2^I$.
Let  $\mathcal{I}=(V,\MC,I,\sigma)$ be an  instance.  
The common item set $\Item(X)$ over a subset $X\subseteq V$
is defined to be  $\Item(X) = \bigcap_{v\in X}\itemmap(v)$.
A  {\em solution\/} to instance $\mathcal{I}$ is defined 
to be a component $X\in \MC$
such that  
\[\mbox{  every component $Y\in \MC$ with $Y\supsetneq X$ satisfies  
 $\Item(Y)\subsetneq \Item(X)$. 
 }\] 
 Let $\MS$ denote the family of all solutions to   instance  $\mathcal{I}$.
 Our aim is to design an efficient algorithm for enumerating all solutions in $\MS$ when
 $\MC$ is transitive in instance $\mathcal{I}$.  

We call  an enumeration algorithm 
 $\mathcal A$ \\
~-   {\em output-polynomial}
if the overall computation time is polynomial with respect to\\
\quad the input and output size; \\
~-  {\em incremental-polynomial}  
if the computation time between the $i$-th output and \\
\quad the $(i-1)$-st output is bounded by a polynomial with respect to \\
\quad  the input size and $i$; and \\
~-  {\em polynomial-delay}  if the delay (i.e., the time between any two consecutive outputs),\\
\quad preprocessing time and postprocessing time are all bounded by a polynomial \\
\quad  with respect to the input size. \\ 
In this paper, we design an  algorithm
that   enumerates all solutions in  $\mathcal S$
by traversing a {\em family tree} over the solutions in  $\mathcal S$,
where the family tree is a tree structure that represents 
a parent-child relationship among solutions.
The following theorem summarizes our main result.

\begin{thm}   \label{thm:main}
  Let $\mathcal{I}=(V,\MC,I,\sigma)$ be 
  an instance on a transitive system $(V,\MC)$ with a volume function $\rho$,
  where $n=|V|$ and $q=|I|$. 
  All $\rho$-positive solutions in $\mathcal S$ 
  to the instance  $\mathcal{I}$ can be enumerated
  in $O\big(q \theta_{2,\mathrm{t}} 
  + (q(n+\theta_{1,\mathrm{t}})+\theta_{\rho,\mathrm{t}})q\delta(V)\big)$
  delay and
  in $O\big((q+n+\theta_{1,\mathrm{s}}+\theta_{2,\mathrm{s}}
  +\theta_{\rho,\mathrm{s}}) n\big)$
  space. 
\end{thm}
The theorem indicates that, when
$\theta_{1,\mathrm{t}}$, $\theta_{2,\mathrm{t}}$, $\theta_{\rho,\mathrm{t}}$ and $\delta(V)$
are bounded by a polynomial of $n$ and $q$, 
all solutions are enumerable in polynomial-delay. 
Similarly, when $\theta_{1,\mathrm{s}}$, $\theta_{2,\mathrm{s}}$ and $\theta_{\rho,\mathrm{s}}$
are bounded by a polynomial of $n$ and $q$, 
the enumeration can be done in polynomial space with respect to the input size. 
Our algorithm in \thmref{main} is a framework that can be
applied to some enumeration problems
over graphs, as will be discussed
in Sections~\ref{sec:app.conn} and \ref{sec:app.subset}.

The paper is organized as follows.
\begin{itemize}
\item We prepare terminologies and notations in \secref{prel}. 
\item In \secref{transitive},
  we present a family-tree based algorithm that
  enumerates all solutions
  in a given instance, along with computational complexity analyses.
  We also show that the algorithm can be used
  to enumerate all components in the transitive system
  of the instance. 
\item
  The proposed algorithm can be applied to
  several problems of enumerating subgraphs 
  that satisfy certain connectivity conditions
  over a given graph.
  In \secref{app.ext.mixed},
  we show how to construct a transitive system
  from a given weighted/unweighted mixed graph
  so that each component in the resulting system
  corresponds to a required subgraph. 
\item In \secref{app.conn},
  we mention a significant application of our algorithm
  to the connector enumeration problem,
  which is used to extract meaningful structure 
  from gene networks~\cite{ASA.2019,SS.2008}.
  Given a graph such that each vertex
  is assigned items,
  the problem asks to enumerate all
  connected induced subgraphs that
  are maximal with respect to the common item set. 
  We show that our algorithm yields the first
  polynomial-delay algorithm for the problem
  even when we require stronger connectivity conditions
  such as $k$-edge/vertex-connectivity. 
\item In \secref{app.subset},
  applying the component enumeration algorithm,
  we obtain polynomial-delay algorithms
  that enumerate all $k$-edge-connected
  (resp., $k$-vertex-connected) induced subgraphs
  and that enumerate all $k$-edge-connected
  (resp., $k$-vertex-connected) spanning subgraphs
  in a given undirected/directed graph 
  for any $k$ (resp., a fixed $k$).
\item Finally  \secref{conc} makes some concluding remarks. 
\end{itemize}

\section{Preliminaries}\label{sec:prel}  

Let $\mathbb{R}$ (resp., $\mathbb{R}_+$) denote the set of reals 
(resp., non-negative reals).
For a function $f: A\to \mathbb{R}$ for a finite subset $A$ 
and a subset $B\subseteq A$,
we let $f(B)$  denote  $\sum_{a\in B}f(a)$. 

For two integers $a$ and $b$, let $[a,b]$ denote the set of integers
$i$ with $a\leq i\leq b$.
For a set $A$ with a total order $<$ over the elements in $A$,
we define a total order $\prec$ over the subsets of $A$ as follows.
For two  subsets $J,K\subseteq A$, 
we denote by $J\prec K$
if the minimum element in $(J\setminus K)\cup(K\setminus J)$ belongs to $J$.
We denote $J\preceq K$ if $J\prec K$ or $J=K$.
Note that $J\preceq K$ holds whenever $J\supseteq K$. 
Let $a_{\max}$ denote the maximum element in $A$. 
Then $J\prec K$ holds for 
$J=\{j_1,j_2,\ldots,j_{|J|}\}$, $j_1<j_2<\cdots<j_{|J|}$ and 
$K=\{k_1,k_2,\ldots,k_{|K|}\}$, $k_1<k_2<\cdots<k_{|K|}$,
if and only if 
 the sequence $(j_1,j_2,\ldots,j_{|J|},j'_{|J|+1},j'_{|J|+2},\ldots, j'_{|A|})$ 
 of length $|A|$  with $j'_{|J|+1}=j'_{|J|+2}=\cdots= j'_{|A|}=a_{\max}$
 is  lexicographically smaller than 
 the  sequence  $(k_1,k_2,\ldots,k_{|K|},k'_{|K|+1},k'_{|K|+2},\ldots, k'_{|A|})$
 of length $|A|$  with $k'_{|K|+1}=k'_{|K|+2}=\cdots=k'_{|A|}=a_{\max}$. 
Hence we see that $\preceq$ is a total order on $2^A$.

We start with an important property on components in a transitive system. 

\begin{lem}  \label{lem:unique} 
Let $(V,\MC)$ be a transitive system. 
For a component $X\in  \MC$ and a superset $Y\supseteq X$,
there is exactly one  component in  $\Max(X;Y)$. 
\end{lem}
\begin{proof}
Since $X\subseteq  Y$, 
$\Max(X;Y)$ contains a $Y$-maximal component $C$. 
For any  component $W\in \MC$ 
with $X\subseteq W\subseteq Y$,  the transitivity of $\MC$
and $X\subseteq C\cap W$ imply $C\cup W\in \MC$,
where   $C\cup W=C$ must hold by the $Y$-maximality of $C$. 
Hence $C$ is unique. 
\end{proof} 

\noindent
For a  component $X\in \MC$  and a superset $Y\supseteq X$,
let $C(X;Y)$ denote the unique component in $\Max(X;Y)$.

Suppose that an instance $(V,\MC,I,\sigma)$ is given. 
To facilitate our aim, we introduce a total order over the items in $I$  
by representing $I$ as a set $[1,q]=\{1,2,\ldots,q\}$ of integers.
For each subset $X\subseteq V$, let 
$\min \Item(X)\in [0,q]$ denote the minimum item in
  $\Item(X)$, where $\min \Item(X)\triangleq 0$ for $\Item(X)=\emptyset$.
For each $i\in [0,q]$, define a family of solutions in $\MS$, 
\[\MS_i\triangleq \{X\in \MS\mid \min \Item(X)=i\}.\] 
Note that $\MS$ is a disjoint union of $\MS_i$, $i\in [0,q]$.
In Section~\ref{sec:trav},
we will design an algorithm that enumerates 
all solutions in $\MS_k$ for any specified integer $k\in [0,q]$.

\section{Enumerating Solutions in Transitive System} 
\label{sec:transitive}

\subsection{Defining Family Tree}\label{sec:defn}

To generate all solutions in $\MS$ efficiently,
we use the idea of family tree, where we first introduce
a parent-child relationship among solutions, which defines
a rooted tree (or a set of rooted trees), and
we traverse each tree starting from the root 
and generating the children of a solution recursively. 
Our tasks to establish such an enumeration algorithm are as follows:
\begin{itemize}
\item[-] 
 Select some solutions from the set $\MS$ of solutions 
 as  the roots, called ``bases;'' 
\item[-]Define the ``parent'' $\pi(S)\in \MS$ of 
each non-base solution $S\in \MS$,
  where the solution $S$ is called a ``child'' of the solution $T=\pi(S)$;  
\item[-] Design an algorithm~A that, given a solution $S\in \MS$,
  returns its parent $\pi(S)$; and  
\item[-] Design an algorithm~B that, given a solution $T\in \MS$, 
generates
  a set $\mathcal{X}$ of components $X\in\MC$ such that 
  $\mathcal{X}$ contains all children of $T$. 
We can test whether each component $X\in\mathcal{X}$ is a child of $T$
by constructing  $\pi(X)$ by algorithm~A and checking  if  $\pi(X)$
 is equal to $T$.
\end{itemize}
Starting from each base, we recursively generate the children of a solution. 
 The complexity of delay-time of the entire algorithm depends on
 the time complexity of algorithms A and B, 
 where $|\mathcal{X}|$ is bounded from above
 by the time complexity of algorithm~B.
 
\subsection{Defining Base}

Let $(V,\MC,I=[1,q],\sigma)$ be an instance on a transitive system. 
We define subsets $V_{\angleb{0}}\triangleq V$
and   $V_{\angleb{i}}\triangleq \{v\in V\mid i\in \itemmap(v)\}$
 for each item $i\in I$.
For each non-empty subset $J\subseteq I$, define subset 
$V_{\angleb{J}}\triangleq \bigcap_{i\in J}V_{\angleb{i}}=\{v\in V\mid J\subseteq \sigma(v)\}$.
For $J=\emptyset$, define $V_{\angleb{J}}\triangleq  V$. 
  For each integer $i\in [0,q]$, define a set of solutions 
 \[\mbox{
 $\MB_i\triangleq \{X\in \Max(V_{\angleb{i}})\mid \min \Item(X)=i\}$,}\]
  and $\MB\triangleq \bigcup_{i\in [0,q]}\MB_i$.
 We call each component in $\MB$ a {\em base}. 
 
\begin{lem}  \label{lem:base}
Let $(V,\MC,I=[1,q],\sigma)$ be an instance on a transitive system.
\begin{enumerate}
\item[{\rm (i)}] For each non-empty set $J\subseteq [1,q]$ or $J=\{0\}$, 
it holds that   $\Max(V_{\angleb{J}})\subseteq \MS$;   
\item[{\rm (ii)}] For each $i\in [0,q]$,  any solution $S\in  \MS_i$ 
  is contained in a base in $\MB_i$;  and  
\item[{\rm (iii)}] $\MS_0=\MB_0$ and  $\MS_q=\MB_q$. 
\end{enumerate}
\end{lem}
\begin{proof} 
(i) 
  Let $X$ be a   component in $\Max(V_{\angleb{J}})$.
  Note that  $J\subseteq \Item(X)$ holds.
  When  $J=\{0\}$ (i.e., $V_{\angleb{J}}=V$), 
   no proper superset of $X$ is a component,   and $X$ is a solution.  
  Consider the case of $\emptyset\neq J\subseteq [1,q]$.  
  To derive a contradiction, assume that $X$ is not a solution; i.e.,
  there is a proper superset $Y$ of $X$ such that $\Item(Y)=\Item(X)$.
  Since $\emptyset\neq J\subseteq \Item(X)=\Item(Y)$, 
  we see that $V_{\angleb{J}}\supseteq Y$.
  This, however, contradicts the $V_{\angleb{J}}$-maximality of $X$. 
   This proves that $X$ is a solution. 
   
 (ii)
  We prove that each solution $S\in  \MS_i$ 
 is contained in a base in $\MB_i$.
 Note that  $i=\min \Item(S)$ holds.
 By definition, it holds that $S\subseteq V_{\angleb{i}}$.
 Let $C\in  \Max(S;V_{\angleb{i}})$ be a solution. 
 Note that $\Item(S)\supseteq \Item(C)$ holds. 
 Since $i\in \Item(C)$ for $i\geq 1$
 (resp., $\Item(C)=\emptyset$ for $i=0$),
 we see that $\min \Item(S)=i=\min \Item(C)$. 
 This proves that $C$ is a base in $\MB_i$.  
 Therefore $S$ is contained in a base $C\in \MB_i$.  
 
 (iii) 
 Let $k\in \{0,q\}$. 
 We see from (i)  that $\Max(V_{\angleb{k}})\subseteq \MS$,
 which implies that  
  $\MB_k
  =\{X\in \Max(V_{\angleb{k}})\mid \min \Item(X)=k\}
  \subseteq \{X\in \MS\mid \min \Item(X)=k\} =\MS_k$.
 We prove that any solution $S\in \MS_k$ is a base in $\MB_k$.
 By (ii), there is a base $X\in \MB_k$ such that $S\subseteq X$,
 which implies that 
 $\Item(S)\supseteq \Item(X)$ and $\min\Item(S)\leq \min\Item(X)$.
 We see that  $ \Item(S)=  \Item(X)$, since 
 $\emptyset=\Item(S)\supseteq \Item(X)$ for $k=0$,
 and $q=\min\Item(S)\leq \min\Item(X)\leq q$ for $k=q$.
Hence $S\subsetneq X$ would contradict that $S$ is a solution.
Therefore 
$S=X\in  \MB_k$,
as required.  
\end{proof}

  Lemma~\ref{lem:base}(iii) tells that all solutions in $\MS_0\cup \MS_q$
  can be found 
by calling oracle $\mathrm{L}_2(Y)$ for $Y=V_{\angleb{0}}=V$ and 
$Y=V_{\angleb{q}}$.
In the following, we consider how to generate 
all solutions in $\MS_k$ for each item $k\in [1,q-1]$.
 
For a notational convenience, 
let  $C(X;i)$ for each item $i\in \Item(X)$ denote 
the component $C(X;V_{\angleb{i}})$ and 
let  $C(X;J)$ for each subset $J\subseteq \Item(X)$
   denote the component $C(X;V_{\angleb{J}})$.

\begin{lem}
  \label{lem:superset}
  Let $(V,\MC,I=[1,q],\sigma)$ be an instance on a transitive system.
  Any two solutions $S,T\in\MS$ such that $S\subseteq T$ satisfy 
   $T=C(S;\Item(T))$.
\end{lem}
\begin{proof}
  Let $T'=C(S;\Item(T))\in \Max(V_{\angleb{\Item(T)}})$.
  Note that $S\subseteq T\subseteq  V_{\angleb{\Item(T)}}$ holds.
  The uniqueness of maximal component  $T'=C(S;\Item(T))$ 
  by \lemref{unique} 
  indicates  $T\subseteq T'$. 
  To derive a contradiction, assume that   $T\subsetneq T'$. 
  By \lemref{base}(i), 
  $T'\in \Max(V_{\angleb{\Item(T)}})$ is a solution.
  Since $T$ and $T'$ are solutions such that $T\subsetneq T'$, 
  it must hold that $\Item(T)\supsetneq \Item(T')$,
  implying that $V_{\angleb{\Item(T)}}\not\supseteq T'$,
  a contradiction.
Therefore we have $T=T'$. 
\end{proof}

\subsection{Defining Parent}

This subsection defines the ``parent'' of a non-base solution. 
For two solutions $S,T\in \MS$,
we say that $T$ is a {\em superset solution} of $S$ 
if  $T\supsetneq S$ and $S,T\in \MS_i$ for some item $i\in [1,q-1]$.
A superset solution $T$ of a solution  $S\in \MS$  is called {\em minimal} 
if no proper subset $Z\subsetneq T$ is a  superset solution of $S$.
Let $S$ be a non-base solution in $\MS_k\setminus \MB_k$
 for some item $k\in [1,q-1]$.
We call a minimal superset solution $T$ of $S$ 
{\em the lex-min solution} of $S$ if 
$\Item(T)  \preceq \Item(T')$ holds
for all minimal superset solutions $T'$ of $S$.

\begin{algorithm}[h]
  \caption{{\sc Parent}$(S)$: 
  Finding the lex-min solution of a  solution $S$ }
  \label{alg:parent_c}
  \begin{algorithmic}[1]
    \Require An instance $(V,\MC,I=[1,q],\itemmap)$  on a transitive system, 
      an item $k\in  [1,q-1]$, and
      a non-base solution $S\in \MS_k\setminus \MB_k$, where $k= \min \Item(S)$.
    \Ensure
     The lex-min solution $T\in \MS_k$ of $S$.
     %
     \State Let $\{k,i_1,i_2,\ldots,i_p\}:=\Item(S)$, where $k<i_1<i_2<\cdots<i_p$;
     \label{line:parent_init}
      \State $J:=\{k\}$;     \Comment{$C(S; k)\supsetneq S$ by $S\not\in \MB_k$}
      \For {{\bf each} integer $j=1,2,\ldots,p$} 
     \label{line:parent_for}
         \If {$C(S;J\cup\{i_j\})\supsetneq S $ } 
         \label{line:parent_eq}
            \State $J:=J\cup\{i_j\}$  
         \EndIf 
      \EndFor; \Comment{$J=\Item(T)$ holds}
     \label{line:parent_endfor}
      \State Output  $T:=C(S;J)$ 
  \end{algorithmic}
\end{algorithm}

\begin{lem}   \label{lem:greedy_minimal} 
Let $(V,\MC,I=[1,q],\sigma)$ be an instance on a transitive system,
  $S\in \MS_k\setminus \MB_k$ be a  non-base  solution for some item $k\in  [1,q-1]$, 
and  $T$ denote the  lex-min solution of $S$. 
Denote $\Item(S)$ by $\{k, i_1,i_2,\ldots,i_p\}$ so that $k<i_1<i_2<\cdots<i_p$.
Then:  
\begin{enumerate}
\item[{\rm (i)}]
For each integer $j\in [1,p]$, 
 $i_j\in \Item(T)$ holds if and only if $C(S;J\cup\{i_j\})\supsetneq S$ holds
for the item set   $J=\Item(T)\cap \{k, i_1,i_2,\ldots,i_{j-1}\}$; and

\item[{\rm (ii)}] 
  {\sc Parent}$(S)$ in \algref{parent_c}
  correctly delivers 
  the lex-min solution of $S$ in $O(q(n+\theta_{\mathrm{1,t}}))$ time
  and $O(q+n+\theta_{\mathrm{1,s}})$ space. 
\end{enumerate}
\end{lem}
\begin{proof}
  (i) By Lemma~\ref{lem:base}(i) and $\min\Item(S)=k$, 
we see that $C(S;J\cup\{i_j\}) \in\MS_k$
for any integer $j\in [1,p]$. \\
\noindent
{\bf  Case~1.} $C(S;J\cup\{i_j\})=S$: 
 For any set $J'\subseteq \{i_{j+1},i_{j+2},\ldots,i_p\}$,
the component $C(S;J\cup\{i_j\}\cup J')$ is equal to $S$ and cannot be
a minimal superset solution of $S$.
This  implies that $i_j\not\in \Item(T)$.\\
\noindent
{\bf  Case~2.}  $C(S;J\cup\{i_j\})\supsetneq S$:
Then $C=C(S;J\cup\{i_j\})$ is a solution by Lemma~\ref{lem:base}(i).
Observe that $k\in J\cup\{i_j\}\subseteq \Item(C)\subseteq \Item(S)$ and
$\min\Item(C)=k$, implying that $C\in \MS_k$ is a superset solution of $S$.
Then $C$ contains
a minimal superset solution $T^*\in \MS_k$ of $S$, where 
$\Item(T^*)\cap[1,i_{j-1}]=\Item(T^*)\cap \{k,i_1,i_2,\ldots,i_{j-1}\}\supseteq 
J= \Item(T)\cap \{k,i_1,i_2,\ldots,i_{j-1}\}=\Item(T)\cap[1,i_{j-1}]$ 
and $i_j \in\Item(T^*)$.
If $\Item(T^*)\cap[1,i_{j-1}]\supsetneq J$ or $i_j\not\in \Item(T)$, 
 then $\Item(T^*)\prec \Item(T)$ would hold, contradicting that $T$ is 
the lex-min solution of $S$.
Hence $\Item(T)\cap[1,i_{j-1}]=J=\Item(T^*)\cap[1,i_{j-1}]$
and    $i_j\in \Item(T)$.

(ii) Based on (i), we can obtain the solution $T$ as follows.
First  we find the item set $\Item(T)$ by applying (i) 
to each integer $j\in[1,p]$,
where we construct subsets 
 $J_0\subseteq J_1\subseteq\cdots \subseteq J_p \subseteq\Item(S)$ such that
  $J_0=\{k\}$ and
\begin{align*}
  J_j&=\left\{
  \begin{array}{ll}
    J_{j-1}\cup\{i_j\} & \textrm{if\ }C(S;J_{j-1}\cup\{i_j\})\supsetneq S,\\
    J_{j-1} & \textrm{otherwise}. 
  \end{array}
    \right.
\end{align*}
Each subset  $J_j$ can be obtained from  subset $J_{j-1}$
by testing whether $C(S;J_{j-1}\cup\{i_j\})\supsetneq S$ holds or not,
where $C(S;J_{j-1}\cup\{i_j\})$ is computable by calling the oracle
$\mathrm{L}_1$.
By (i), we have $J_j=\Item(T)\cap\{k,i_1,\dots,i_j\}$,
and in particular, $J_p=\Item(T)$ holds.
Next we compute the component $C(S;J_p)$ by calling the oracle 
$\mathrm{L}_1(S,V_{\angleb{J_p}})$, 
where   $C(S;J_p)$ is equal to the solution $T$ by \lemref{superset}.
The above algorithm is described 
as  algorithm {\sc Parent}$(S)$ in \algref{parent_c}.

Let us mention critical parts in terms of time complexity analysis. 
In \lineref{parent_init}, it takes $O(qn)$ time to compute $\Item(S)$.
The for-loop from \lineref{parent_for} to
\ref{line:parent_endfor} is repeated $O(q)$ times.
In \lineref{parent_eq}, the oracle $\mathrm{L}_1(S,V_{\angleb{J\cup\{i_j\}}})$
is called to obtain a component $Z=C(S;J\cup\{i_j\})$
and whether $S=Z$ or not is tested.
This takes $O(\theta_{\mathrm{1,t}}+n)$ time.
The overall running time is $O(q(n+\theta_{\mathrm{1,t}}))$.
It takes $O(q)$ space to store $\Item(S)$ and $J$,
and $O(n)$ space to store $S$ and $Z$.
An additional $O(\theta_{\mathrm{1,s}})$
space is needed for the oracle $\mathrm{L}_1$. 
\end{proof}

For each item  $k\in [1,q-1]$,
 we define the {\em parent} $\pi(S)$ of a non-base solution 
 $S\in \MS_k\setminus \MB_k$ to be the lex-min solution  of $S$,
  and   define  a {\em child} of  a solution $T\in \MS_k$
  to be a non-base solution $S\in \MS_k\setminus \MB_k$ such that $\pi(S)=T$.

\subsection{Generating Children}

This subsection shows how to construct a family $\mathcal{X}$ of
components for a given solution $T$ so that $\mathcal{X}$ contains
all children of $T$. 

\begin{lem}   \label{lem:child_candidate}  
Let $(V,\MC,I=[1,q],\sigma)$ be an instance on a transitive system
and   $T\in \MS_k$ be a solution for some item $k\in  [1,q-1]$.
Then:
 \begin{enumerate}
 \item[{\rm (i)}] 
   Every child $S$ of $T$   
   satisfies  $[k+1,q]\cap (\Item(S)\setminus \Item(T)) \neq\emptyset$
   and is a component in $\Max(T\cap V_{\angleb{j}})$ 
   for any item  $j\in[k+1,q]\cap (\Item(S)\setminus \Item(T))$;
 \item[{\rm (ii)}] 
The family of children $S$ of $T$ is
equal to the disjoint collection of families 
$\mathcal{C}_j =
   \{ C\in\Max(T\cap V_{\angleb{j}})\mid 
        k= \min \Item(C),
        j=\min\{i\mid i\in [k+1,q]\cap (\Item(C)\setminus \Item(T))\},
        T=${\sc Parent}$(C)\}$ over all items  
        $j\in[k+1, q]\setminus \Item(T)$; and 
 \item[{\rm (iii)}] 
  The set of all children of $T$ can be constructed in 
$O\big(q \theta_{2,\mathrm{t}} + q^2(n+\theta_{1,\mathrm{t}})\delta (T) \big)$ 
time and
   $O(q+n+\theta_{1,\mathrm{s}}+\theta_{2,\mathrm{s}})$ space. 
 \end{enumerate}
\end{lem}
\begin{proof} 
(i) 
Note that $[0,k]\cap \Item(S)=[0,k]\cap \Item(T)=\{k\}$ since $S,T\in \MS_k$. 
Since $S\subseteq T$ are both solutions,
$\Item(S)\supsetneq \Item(T)$.
Hence 
$[k+1,q]\cap (\Item(S)\setminus \Item(T)) \neq\emptyset$.  
Let $j$ be an arbitrary item in $[k+1,q]\cap (\Item(S)\setminus \Item(T))$.
Since $S\subseteq T\cap V_{\angleb{j}}$, 
it holds that 
$\Max(S;T\cap V_{\angleb{j}})\neq\emptyset$.
 
Let $C$ be a   $(T\cap V_{\angleb{j}})$-maximal component in $\Max(S;T\cap V_{\angleb{j}})$.
It suffices to show that $C=S$. 
Note that  $S\subseteq C\subseteq T$, 
$\Item(S)\supseteq \Item(C)\supseteq \Item(T)$ and 
 $k=\min \Item(S)=\min \Item(T)$ implies $ \min \Item(C)= k$.
 
 We show that $C\in \MS$, which implies $C\in \MS_k$. 
 Note that $j \in \Item(C)\setminus \Item(T)$, and $C\subsetneq T$.
 Assume that $C$ is not a solution; i.e., there is a solution  $C^*\in \MS$
 such that $C\subsetneq C^*$ and  $\Item(C)=\Item(C^*)$,
 where $j\in \Item(C)=\Item(C^*)$ means that $C^*\subseteq V_{\angleb{j}}$.
 Hence  $C^*\setminus T\neq\emptyset$ 
 by the $(T\cap V_{\angleb{j}})$-maximality 
 of $C$.
 Since  $C,C^*,T\in \MC$ and $C\subseteq C^*\cap T$, 
we have $C^*\cup T\in \MC$ by the transitivity. 
We also see that 
$\Item(C^*\cup T)=\Item(C^*)\cap \Item(T)=\Item(C)\cap \Item(T)=\Item(T)$.
 This, however, contradicts that $T$ is a solution,
 proving that $C\in \MS_k$.  
 If  $S\subsetneq C$, then $S\subsetneq C\subsetneq T$ would hold
  for $S,C,T\in \MS_k$,
 contradicting that $T$ is a minimal superset solution of $S$.
 Therefore $S=C$. 
 
 (ii)  
 By (i), the family $\mathcal{S}_T$ of children of $T$ is 
 contained in the family of 
  $(T\cap V_{\angleb{j}})$-maximal components 
  over all items $j\in [k+1,q]\cap  \Item(T)$.
Hence $\mathcal{S}_T
    =\cup_{j\in [k+1,q]\cap \Item(T)}\{C\in \Max(T\cap V_{\angleb{j}})
    \mid T=${\sc Parent}$(C)\}$.
Note that 
  if a subset $S\subseteq V$ is a child of $T$, then 
 $k= \min \Item(S)$ and 
 $S\in \Max(T\cap V_{\angleb{j}})$ for all items 
$j\in [k+1,q]\cap (\Item(S)\setminus \Item(T))$.
Hence we see that  $\mathcal{S}_T$ is
equal to the disjoint collection of families 
$\mathcal{C}_j =
   \{ C\in\Max(T\cap V_{\angleb{j}})\mid 
        k= \min \Item(C),
        j=\min\{i\mid i\in [k+1,q]\cap (\Item(C)\setminus \Item(T))\},
        T=${\sc Parent}$(C)\}$ over all items  
        $j\in[k+1, q]\setminus \Item(T)$. 
 
\begin{algorithm}[h]
 \caption{{\sc Children}$(T,k)$: Generating all children}\label{alg:children}
\begin{algorithmic}[1]
\Require  An instance $(V,\MC,I,\itemmap)$, an item $k\in   [1,q-1]$ and   
a solution $T\in \MS_k$. 
\Ensure All children of $T$, each of which is output whenever it is generated.
\For{{\bf each} item  $j\in[k+1, q]\setminus \Item(T)$} 
 \State Compute $\Max(T\cap V_{\angleb{j}})$; 
    \For{{\bf each} component $C\in\Max(T\cap V_{\angleb{j}})$} 
        \If {$k= \min \Item(C)$, 
        $j=\min\{i\mid i\in [k+1,q]\cap (\Item(C)\setminus \Item(T))\}$ \par 
          \hskip\algorithmicindent and $T=${\sc Parent}$(C)$ 
          }
                   \State Output $C$ as one of the children of $T$ 
              \EndIf    
    \EndFor
\EndFor 
  \end{algorithmic}
\end{algorithm}

 (iii)  
Based on (ii), we obtain an algorithm described in Algorithm~\ref{alg:children}.
We analyze the time and space complexities of the algorithm.
Note that   $T$ may have no children.
The outer for-loop from line 1 to 10 is repeated $O(q)$ times.
Computing $\MC(T\cap V_{\angleb{j}})$ in line 2
takes $\theta_{2,\mathrm{t}}$ time by calling the oracle L$_2$.  
The inner for-loop from line 3 to 7 is repeated
at most $\delta(T\cap V_{\angleb{j}})$ times for each $j$,
and 
the most time-consuming part of the inner for-loop is 
  algorithm \textsc{Parent}$(S)$ in line 4, which takes
   $O(q(n+\theta_{1,\mathrm{t}}))$ time
  by  \lemref{greedy_minimal}(ii).
  Recall that $\delta$ is a non-decreasing function. 
Then the running time of  algorithm \textsc{Children}$(T,k)$ is evaluated by
\[
O\Big(q\theta_{2,\mathrm{t}} 
+ q(n+\theta_{1,\mathrm{t}})\sum_{j\in[k+1,q]\setminus\Item(T)}
\delta(T\cap V_{\angleb{j}})\Big)
=
O\big(q \theta_{2,\mathrm{t}} + q^2(n+\theta_{1,\mathrm{t}})\delta (T) \big). 
\]

For the space complexity,
we do not need to share the space between
iterations of the outer for-loop from line 1 to 8.
In each iteration,
we use the oracle L$_2$ and  algorithm \textsc{Parent}$(S)$, whose space complexity 
is $O(q+n+\theta_{1,\mathrm{s}})$   by  \lemref{greedy_minimal}(ii).
Then  algorithm \textsc{Children}$(T,k)$ uses
$O(q+n+\theta_{1,\mathrm{s}}+\theta_{2,\mathrm{s}})$ space.
\end{proof}

\subsection{Traversing Family Tree}
\label{sec:trav}

We are ready to describe an entire algorithm for enumerating
solutions in $\MS_k$ 
for a given integer $k\in [0,q]$. 
We first compute the component set $\Max(V_{\angleb{k}})$. 
We next compute the family $\MB_k~(\subseteq \Max(V_{\angleb{k}}))$ of bases
by testing whether $k=\min \Item(T)$ or not,
where $\MB_k\subseteq \MS_k$.
When $k=0$ or $q$, we are done with $\MB_k=\MS_k$ 
by Lemma~\ref{lem:base}(iii). 
Let  $k\in [1,q-1]$.
Suppose that we are given a solution $T\in \MS_k$.
We find all the children of $T$ by
{\sc Children}$(T,k)$ in  Algorithm~\ref{alg:children}. 
By applying Algorithm~\ref{alg:children}
to a newly found child recursively,
we can find all solutions in $\MS_k$.

When no child is found to a given solution $T\in \MS_k$, 
we may need to go up to an ancestor by traversing
recursive calls $O(n)$ times before we generate the next solution.
This would result in time delay of $O(n\alpha)$, where $\alpha$ denotes
the time complexity required for a single run of
{\sc Children}$(T,k)$. 
To improve the delay to $O(\alpha)$, we employ the  
{\em alternative output method\/}~\cite{U.2003}, where  
we output the children of $T$ after (resp., before)
generating all descendants when the depth of the recursive call to $T$
is an even (resp., odd) integer. 

Assume that a  volume function   $\rho: 2^V\to \mathbb{R}$ 
is given. 
An algorithm that enumerates all  $\rho$-positive solutions in $\MS_k$
is described in Algorithm~\ref{alg:enumalg}
and Algorithm~\ref{alg:enumdesc}.

\begin{algorithm}[h]
  \caption{An algorithm to enumerate $\rho$-positive solutions in $\MS_k$ 
  for a given $k\in [0,q]$}
  \label{alg:enumalg}
  \begin{algorithmic}[1]
    \Require An instance $(V,\MC,I=[1,q],\itemmap)$ on a transitive system, 
     and an item $k\in  [0,q]$
    \Ensure The set $\MS_k$ of   solutions to $(V,\MC,I,\itemmap)$ 
    \State Compute $\Max(V_{\angleb{k}})$; $d:=1$; 
    \label{line:enum_vk}
    \For{{\bf each} $T\in\Max(V_{\angleb{k}})$}
    \label{line:enum_for}
    \If{$k=\min \Item(T)$ (i.e., $T\in \MB_k$) and $\rho(T)>0$}
    \label{line:enum_if}
         \State  Output $T$; 
         \If {$k\in [1,q-1]$}
         \State {\sc Descendants}$(T,k,d+1)$
         \EndIf
    \EndIf
    \EndFor
    \label{line:enum_endfor}
  \end{algorithmic}
\end{algorithm}
  
\begin{algorithm}[h]
  \caption{{\sc Descendants}$(T,k,d)$: Generating all $\rho$-positive descendant solutions}
  \label{alg:enumdesc}
  \begin{algorithmic}[1]
\Require An instance $(V,\MC,I,\itemmap)$, $k\in   [1,q-1]$, 
a solution $T\in \MS_k$, 
the current depth $d$ of recursive call of {\sc Descendants}, and
a volume function $\rho:2^V\to\mathbb{R}$
\Ensure All $\rho$-positive descendant solutions of $T$ in $\MS_k$ 
 \For{{\bf each} item $j\in[k+1, q]\setminus \Item(T)$}
 \label{line:desc_for}
 \State Compute $\Max(T\cap V_{\angleb{j}})$;
 \label{line:desc_max}
    \For{{\bf each} component $S\in\Max(T\cap V_{\angleb{j}})$}
      \label{line:desc_inner_for}
      \If {$k= \min \Item(S)$,
               $j=\min\{i\mid i\in [k+1,q]\cap (\Item(S)\setminus \Item(T))\}$, \par
        \hskip\algorithmicindent $T= ${\sc Parent}$(S)$  (i.e., $S$ is a child of $T$), and $\rho(S)>0$  }
            \If{$d$ is odd}
              \State Output $S$
            \EndIf;   
            \State {\sc Descendants}$(S,k,d+1)$; 
            \If{$d$ is even}
                \State Output $S$
            \EndIf 
    \EndIf
   \EndFor
 \label{line:desc_inner_endfor}
 \EndFor   
 \label{line:desc_endfor}
  \end{algorithmic}
\end{algorithm}

\begin{lem}   \label{lem:main_poly} 
Let $(V,\MC,I=[1,q],\sigma)$ be an instance on a transitive system. 
For each $k\in [0,q]$, all $\rho$-positive solutions in $\MS_k$ can be enumerated
in
$O\big(q\theta_{2,\mathrm{t}} + (q(n+\theta_{1,\mathrm{t}})
+\theta_{\rho,\mathrm{t}})q\delta(V_{\angleb{k}})\big)$
 delay and 
$O\big((q+n+\theta_{1,\mathrm{s}}+\theta_{2,\mathrm{s}}
+\theta_{\rho,\mathrm{s}}) n\big)$
space.
\end{lem} 
\begin{proof}
  Let $T\in\MS_k$ be a solution such that $\rho(T)\leq 0$. 
  In this case, $\rho(S)\le\rho(T)\leq 0$ holds
  for all descendants $S$ of $T$ since $S\subseteq T$.
  Then we do not need to make recursive calls for such $T$. 
  
  We analyze the time delay. 
  Let $\alpha$ denote the time complexity required for
  a single run of {\sc Children}$(T,k)$.
  By \lemref{child_candidate}(ii) and $\delta(T)\le\delta(V_{\angleb{k}})$,
  we have $\alpha=O\big(q\theta_{2,\mathrm{t}} 
  + q^2(n+\theta_{1,\mathrm{t}})\delta(V_{\angleb{k}})\big)$. 
  In \algref{enumalg} and {\sc Descendants},
  we also need to compute $\rho(S)$ for all child candidates $S$.
  The complexity is $O(q\delta(V_{\angleb{k}})\theta_{\rho,\mathrm{t}})$
  since $\rho(S)$ is called at most $q\delta(V_{\angleb{k}})$ times.
  Hence we see that the time complexity of  \algref{enumalg} 
  and {\sc Descendants}
  without including recursive calls is $O(\alpha+q\delta(V_{\angleb{k}})\theta_{\rho,\mathrm{t}})$. 
  
  From \algref{enumalg} and {\sc Descendants}, we observe: \\
  (i) When $d$ is odd, the solution $S$ for any call {\sc Descendants}$(S,k,d+1)$
  is output\\
  ~~ immediately before  {\sc Descendants}$(S,k,d+1)$ is executed; and \\
  (ii) When $d$ is even, the solution $S$ for any call
   {\sc Descendants}$(S,k,d+1)$
  is output\\
  ~~  immediately after  {\sc Descendants}$(S,k,d+1)$ is executed. \\ 
  Let $m$ denote the number of all calls of {\sc Descendants} during a whole
  execution of  \algref{enumalg}.
  Let $d_1=1,d_2,\ldots,d_m$ denote the sequence of depths $d$
  in each {\sc Descendants}$(S,k,d+1)$ of the $m$ calls.
  Note that $d=d_i$ satisfies (i) when $d_{i+1}$ is odd and $d_{i+1}=d_i+1$, 
  whereas $d=d_i$ satisfies (ii) when  $d_{i+1}$ is even and $d_{i+1}=d_i-1$.
  Therefore we easily see that during three consecutive calls 
  with depth $d_i$, $d_{i+1}$ and $d_{i+2}$,
  at least one solution will be output.
  This implies that the time delay for outputting a solution is
   $O(\alpha+q\delta(V_{\angleb{k}})\theta_{\rho,\mathrm{t}})$.

  We analyze the space complexity. 
  Observe that the number of calls  {\sc Descendants} whose executions
   are not finished  
  during an execution of \algref{enumalg} is the depth $d$
  of the current call {\sc Descendants}$(S,k,d+1)$. 
  In \algref{enumdesc}, 
  $|T|+d\leq n+1$ holds initially,
  and 
   {\sc Descendants}$(S,k,d+1)$ is called
  for a nonempty subset $S\subsetneq T$, where $|S|<|T|$.
  Hence $|S|+d\leq n+1$ holds when  {\sc Descendants}$(S,k,d+1)$ is called.
Then  \algref{enumalg} can be implemented to run 
 in $O(n(\beta+\theta_{\rho,\mathrm{s}}))$ space, where
  $\beta$ denotes the space required for a single run of {\sc Children}$(T,k)$.
  We have $\beta=O(q+n+\theta_{1,\mathrm{s}}+\theta_{2,\mathrm{s}})$
  by \lemref{child_candidate}(ii).
  Then the overall space complexity is 
  $O\big((q+n+\theta_{1,\mathrm{s}}+\theta_{2,\mathrm{s}}+\theta_{\rho,\mathrm{s}}) n\big)$.   
\end{proof}

The volume function is introduced to impose a condition on the output solutions. 
For example,
when $\rho(X)=|X|-p$ for a constant $p$,
all solutions $X\in\MS_k$ with $|X|\ge p+1$ will be output.
In particular, all solutions in $\MS_k$ will be output for $p\le0$. 
In this case, we have 
$\theta_{\rho,\mathrm{t}}=\theta_{\rho,\mathrm{s}}=O(n)$,
and thus the delay is $O\big(q\theta_{2,\mathrm{t}} 
+ q^2(n+\theta_{1,\mathrm{t}})\delta(V_{\angleb{k}})\big)$ and
the space is 
$O\big((q+n+\theta_{1,\mathrm{s}}+\theta_{2,\mathrm{s}}) n\big)$. 

\thmref{main} is immediate from
\lemref{main_poly}
since $\delta(V_{\angleb{k}})\le\delta(V)$ holds
by our assumption
that $\delta(Y)\le\delta(X)$ for subsets $Y\subseteq X\subseteq V$. 

\subsection{Enumerating Components}\label{sec:enum}

This section shows that our algorithm in the previous section can enumerate
all components in a given transitive system  $(V,\MC)$ with $n=|V|\geq 1$. 
For this, we construct an instance $\mathcal{I}=(V,\MC,I=[1,n],\varphi)$
as follows.
Denote $V$ by $\{v_1,\dots,v_n\}$.
We set $I=[1,n]$
and define a function $\varphi:V\to 2^I$
to be $\varphi(v_k)\triangleq I\setminus\{k\}$ for each element $v_k\in V$.
For each subset $X\subseteq V$, let  $\Ind(X)$ denote the set of indices $i$ 
of elements $v_i\in X$; i.e.,   $\Ind(X)=\{i\in [1,n]\mid v_i\in X\}$,
and  $I_\varphi(X)\subseteq [1,n]$ 
denote the common item set over $\varphi(v)$, $v\in X$;
i.e.,   $I_\varphi(X) = \bigcap_{v\in X}\varphi(v)$.
Observe that $I_\varphi(X)=I\setminus\Ind(X)$. 

\begin{lem}
  \label{lem:compsol}
  Let $(V=\{v_1,\dots,v_n\},\MC)$ be a  transitive system with $n\geq 1$.
  The family $\MC$ of all components   is equal 
  to the family $\MS$ of all solutions
  in the instance $(V,\MC,I=[1,n],\varphi)$. 
\end{lem}
\begin{proof} Since any solution $S\in \MS$ is a component, it holds that 
$\MC\supseteq \MS$.
We prove that $\MC\subseteq \MS$.
  Let $X\in\MC$. 
  For any superset $Y\supsetneq X$,
  it holds that $I_{\varphi}(Y)=I\setminus \Ind(Y)\subsetneq I\setminus \Ind(X)=I_{\varphi}(X)$.
  The component $X$ is a solution in $(V,\MC,I,\varphi)$
  since no superset of $X$ has the same common item set as $X$. 
\end{proof}

Since the family  $\MC$ of components is equal to the family $\MS$ of solutions
to the instance $\mathcal{I}=(V,\MC,I,\varphi)$   by \lemref{compsol}, 
we can enumerate all components in $(V,\MC)$ by
running our algorithm on the instance $\mathcal{I}$. 
By $|I|=n$, we have the following corollary to \thmref{main}. 

\begin{cor}\label{cor:comp} 
  Let $(V,\MC)$ be a  transitive system with $n=|V|\geq 1$
  and a volume function $\rho$.
  All $\rho$-positive components in $\MC$ can be enumerated 
  in $O\big( n\theta_{2,\mathrm{t}} 
     + (n^2+n\theta_{1,\mathrm{t}}
      +\theta_{\rho,\mathrm{t}})n \delta(V)\big)$   delay and 
  $O\big((n+\theta_{1,\mathrm{s}}+\theta_{2,\mathrm{s}}
  +\theta_{\rho,\mathrm{s}}) n\big)$
  space.
\end{cor} 


\section{Transitive System in Mixed Graph with Meta-weight Function} 
\label{sec:app.ext.mixed}

Our enumeration algorithm in a transitive system can be
applied to several  problems of enumerating subgraphs 
that satisfy certain types of connectivity requirements 
over a given graph. 
To treat these applications universally,
this subsection presents a general method of constructing a transitive system
based on a mixed graph and a weight function on elements in the graph.
 
\subsection{Meta-weight Function in Mixed Graph} 

Let $M$ be a {\em mixed graph}, which is defined to be
 a graph that may contain undirected edges and directed edges.
In this paper, $M$ may have multiple edges but no self-loops.
Let $V(M)$, $\vec{E}(M)$ and $\overline{E}(M)$   denote
the sets of vertices, directed edges and undirected edges, respectively.
Let $E(M)\triangleq \vec{E}(M)\cup \overline{E}(M)$.
Let $n=|V(M)|$ and $m=|E(M)|$. 
For a vertex subset $X\subseteq V$, 
let $M[X]$ denote the subgraph  induced from $M$ by $X$.
For a subset $X\subseteq  V(M)\cup E(M)$,
let $V(X)$ denote the set of vertices in $X\cap V(M)$ and
the end-vertices of edges in $X\cap E(M)$. 
For two vertices $u,v\in V(M)$, let \\
~~  $\vec{E}(u,v)$ denote the set of  directed edges from $u$ to $v$, \\
~~  $\overline{E}(u,v)$ denote the set of 
 undirected edges between $u$ and $v$ in $M$,  and \\
~~ $E(u,v)\triangleq \vec{E}(u,v)\cup \overline{E}(u,v)$. \\
 For two non-empty subsets $X,Y\subseteq V(M)$, let \\
~~ $\vec{E}(X,Y)\triangleq \bigcup_{u\in X,v\in Y}\vec{E}(u,v)$, 
$\overline{E}(X,Y)\triangleq \bigcup_{u\in X,v\in Y}\overline{E}(u,v)$ and \\
~~ $E(X,Y)\triangleq \bigcup_{u\in X,v\in Y}E(u,v)$. \\
For two vertices $s,t\in V(M)$,
an {\em $s,t$-cut $C$} is defined to be an ordered pair $(S,T)$
of disjoint subsets $S,T\subseteq V(M)$ such that 
$s\in S$ and $t\in T$, and the element set 
$\varepsilon(C)$ of $C$
(or $\varepsilon(S,T)$ of $(S,T)$) 
is defined to be a union $F\cup R$ of
the edge subset $F=E(S,T)$  and
the vertex subset $R=V(M)\setminus (S\cup T)$,  where
$R=\emptyset$ is allowed.

We define a {\em meta-weight function} on $M$ to be
$\omega: 2^{V(M)\cup E(M)}\times (V(M)\cup E(M))\to \mathbb{R}_+$.
For each subset $X\in 2^{V(M)\cup E(M)}$,
we define
the function $\omega_X: V(M)\cup E(M)\to \mathbb{R}_+$ induced
from $\omega$ by $X$ so that 
 $\omega_X(a)=\omega(X,a)$ for each element $a\in V(M)\cup E(M)$. 
We call $\omega$ {\em monotone}
 if every two subsets $X\subseteq Y\subseteq V(M)$ satisfy 
\[
\mbox{ 
$\omega_{Y}(a)\geq \omega_X(a)$ for each element $a\in V(M)\cup E(M)$.}\]
For  two vertices $s,t\in V(M)$ 
and a subset $X\subseteq V(M)\cup E(M)$, define 
\[ \mu(s,t;X)\triangleq \min\{\omega_X(\varepsilon(C))\mid
\mbox{$s,t$-cuts $C=(S,T)$ in $M$}\}. \]
We call a subset $X\subseteq V(M)\cup E(M)$ {\em $k$-connected}
if   $|V(X)|=1$ or $\mu(u,v;X)\geq k$ for each pair of vertices $u,v\in V(X)$.

\begin{lem}   \label{lem:connectivity}
Let $(M,\omega)$ be a mixed graph with a monotone meta-weight function,
and $k\geq 0$.  
For any two $k$-connected subsets $X,Y\subseteq V(M)\cup E(M)$ such that 
$\omega_{X\cap Y}(V(X\cap Y))\geq k$, 
the subset  $X\cup Y$  is $k$-connected.
 \end{lem}
\begin{proof}
To derive a contradiction, assume that $X\cup Y$  is not $k$-connected;
i.e., $|V(X\cup Y)|\geq 2$ and 
some vertices $s,t\in V(X\cup Y)$ admit 
an $s,t$-cut $C=(S,T)$ with $\omega_{X\cup Y}(\varepsilon(C))<k$.
By the monotonicity of $\omega$, it holds that 
$\omega_{X\cup Y}(a)\geq \omega_X(a), \omega_Y(a)$
 for any element $a\in V(M)\cup E(M)$.
Hence $\omega_{X\cup Y}(\varepsilon(C))<k$ implies 
$\omega_{X}(\varepsilon(C))<k$ and $\omega_{Y}(\varepsilon(C))<k$.
Since each of $X$ and $Y$ is $k$-connected, we see that  
  neither of $s,t\in V(X)$ and $s,t\in V(Y)$  occurs.
  Without loss of generality assume that 
  $s\in V(X\setminus Y)$ and $t\in V(Y\setminus X)$.
  If some vertex $v\in V(X\cap Y)$ belongs to $T$ (resp., $S$),
  then $C$ would be 
  an $s,v$-cut with $s,v\in V(X)$ (resp., $v,t$-cut with $v,t\in V(Y)$),
  contradicting the $k$-connectivity of $X$ (resp., $Y$).
  Hence for the set $R=V(M)\setminus (S\cup T)$,
  it holds $V(X\cap Y)\subseteq R$.  
By the assumption of $X\cap Y$, the non-negativity and 
the monotonicity of $\omega$, 
  we have $k\leq \omega_{X\cap Y}(V(X\cap Y)) \leq \omega_{X\cap Y}(R)
  \leq \omega_{X\cup Y}(R)\leq \omega_{X\cup Y}(\varepsilon(C))$.
  This, however, contradicts $\omega_{X\cup Y}(\varepsilon(C))<k$.
\end{proof}

For a mixed graph $(M,\omega)$  with a   meta-weight function
and a real  $k\geq 0$,  
let $\MC(M,\omega,k)\subseteq 2^{V(M)\cup E(M)}$ denote 
the family of $k$-connected subsets $X\subseteq V(M)\cup E(M)$ with
$\omega_X(V(X))\geq k$.

\begin{lem}   \label{lem:transitive}
  For a mixed graph $(M,\omega)$  with a monotone meta-weight function
  and
 a real  $k\geq 0$, let $\MC=\MC(M,\omega,k)$.
Then $ \MC$ is   transitive.  
 \end{lem}
\begin{proof}
Let  $Z,X,Y\in \MC$ such that $Z\subseteq X\cap Y$, where
 $\omega_{X\cup Y}(V(X\cup Y))\geq \omega_{X\cup Y}(V(Z))
 \geq \omega_{Z}(V(Z))\geq k$.
 By $\omega_{Z}(V(Z))\geq k$ and Lemma~\ref{lem:connectivity}, $X\cup Y$ is $k$-connected.
 Since $\omega_{X\cup Y}(V(X\cup Y))\geq k$, it holds that $X\cup Y\in \MC$.  
 Therefore $\MC$ is transitive.
\end{proof} 

\subsubsection{Construction of Monotone Meta-weight Functions}
\label{sec:app.ext.const}

This part shows a concrete method of constructing a monotone meta-weight function
from a mixed graph with a standard weight function on the vertex and edge sets.
We also present how to construct oracles L$_1$ and L$_2$
that are required when we apply the enumeration algorithm in \secref{trav}
to the corresponding transitive system. 

Let $M$ be a mixed graph and $w:V(M)\cup E(M)\to \mathbb{R}_+$
 be a weight function.
We define a {\em coefficient function} to be 
  $\gamma=(\alpha,\overline{\alpha},\alpha^+,\alpha^-,\beta)$  that consists of  
  functions \\
\quad 
$\alpha :  E(M)\to \mathbb{R}_+$, 
$\overline{\alpha} : \overline{E}(M)\to \mathbb{R}_+$, 
$\alpha^+,\alpha^- : \vec{E}(M)\to \mathbb{R}_+$, and \\
\quad  $\beta: V(M)\cup E(M) \to \mathbb{R}_+$. \\
We call $\gamma$ {\em monotone} if \\
\quad $1\geq \alpha(e)\geq \overline{\alpha}(e)\geq \beta(e)$
     for each undirected edge $e\in \overline{E}(M)$, \\
\quad  
$1\geq \alpha(e)\geq \alpha^+(e)\geq \beta(e)$
  for each directed edge $e\in \vec{E}(M)$;   \\
\quad 
$1\geq \alpha(e)\geq \alpha^-(e)\geq \beta(e)$ 
 for each directed edge $e\in \vec{E}(M)$; 
and \\
\quad 
$1\geq \beta(v)$ for each vertex $v\in V(M)$. \\
We call a tuple $(M,w,\gamma)$ a {\em system},
and define a  meta-weight function
$\omega: 2^{V(M)\cup E(M)}\times (V(M)\cup E(M))\to \mathbb{R}_+$
 to the system 
so that, for each subset $X\subseteq V(M)\cup E(M)$,   
$\omega_X:V(M)\cup E(M) \to\mathbb{R}_+$ is given by \\
\[
\omega_X(v)= \left \{\begin{array} {r l}
            w(v) & \mbox{if $v\in V(X)$,} \\ 
            \beta(v)w(v)  & \mbox{if $v\in V(M)\setminus V(X)$,}
              \end{array}
\right. \]
\[
\omega_X(e)= \left \{\begin{array} {r l}
            w(e) & \mbox{if $e \in E(M)\cap X$,} \\ 
            \alpha(e)w(e) & \mbox{if $e \in E(V(X),V(X))\setminus X$,} \\ 
            \overline{\alpha}(e)w(e) & 
                       \mbox{if $e\in \overline{E}(V(X),V(M)\setminus V(X))$,} \\ 
            \alpha^+(e)w(e) & \mbox{if $e\in \vec{E}(V(X),V(M)\setminus V(X))$,} \\ 
            \alpha^-(e)w(e) & \mbox{if $e\in \vec{E}(V(M)\setminus V(X), V(X))$,} \\ 
            \beta(e)w(e)  & \mbox{if $e \in E(V\setminus V(X),V\setminus V(X))$.}
              \end{array}
\right. \]
We call a system $(M,w,\gamma)$ {\em monotone} if $\gamma$ is monotone. 

\begin{lem}   \label{lem:monotone}
For a monotone system $(M,w,\gamma)$, the corresponding meta-weight function
$\omega: 2^{V(M)\cup E(M)}\times (V(M)\cup E(M))\to \mathbb{R}_+$ 
  is monotone. 
 \end{lem}
\begin{proof} 
  Let $X\subseteq  Y\subseteq V(M)\cup E(M)$,
  where $V(X)\subseteq V(Y)$ holds. 
  It suffices to show that $\omega_Y(a)\geq \omega_X(a)$ 
  for any element $a\in V(M)\cup E(M)$.
  For each vertex $v\in V(M)$, we see that  
  $\omega_Y(v)
  =\omega_X(v)+|\{v\}\cap (V(Y)\setminus V(X))|(1-\beta(v))w(v)
  \geq \omega_X(v)$. 
    For each edge $e\in E(M)$ with end-vertices $u$ and $v$, we see that
  (i)
    $\omega_Y(e)=\omega_X(e)+(1-\alpha)|\{e\}\cap(Y\setminus X)|w(e)\ge\omega_X(e)$
    if $u,v\in V(X)$;
  and
(ii)
$\omega_Y(e)=\omega_X(e)+\Delta|\{u,v\}\cap (V(Y)\setminus V(X))| w(e)\geq \omega_X(e)$
  otherwise,
where  
$\Delta$ is one of
$1-\overline{\alpha}(e)$, $1- \alpha^+(e)$, $1-\alpha^-(e)$, 
$\alpha(e)-\overline{\alpha}(e)$, $\alpha(e)- \alpha^+(e)$,
$\alpha(e)-\alpha^-(e)$, 
$(\alpha(e)-\beta(e))/2$,
$\overline{\alpha}(e)-\beta(e)$,
$\alpha^+(e)-\beta(e)$, $\alpha^-(e)-\beta(e)$ and
$(1-\beta(e))/2$.
\end{proof}

For a system $(M,w,\gamma)$ on a mixed graph $M$ with  
 $n$ vertices and $m$ edges 
and a real $k\geq 0$, let $\mathrm{tm}(n,m,k)$ and $\mathrm{sp}(n,m,k)$ respectively 
 denote the time and space complexities for testing if $\mu(u,v;X)<k$
holds or not for two vertices $u,v\in V(M)$ and
 a subset $X\subseteq V(M)\cup E(M)$.

\begin{lem}   \label{lem:flow}   
For a monotone system $(M,w,\gamma)$, let  $\omega$ be the corresponding 
 monotone meta-weight function.
\begin{enumerate}
\item[{\rm (i)}]   
$\mathrm{tm}(n,m,k)=O(mn\log n)$ and $\mathrm{sp}(n,m,k)=O(n+m)$; and 
\item[{\rm (ii)}]  
 Let  $X\subseteq Y\subseteq V(M)\cup E(M)$ be non-empty subsets such that 
 $\omega_X(V(X))\geq k$ and $\mu(u,u';Y)\geq k$ for all vertices $u,u'\in V(X)$.  
Given a vertex $t\in V(Y)\setminus V(X)$, 
 whether  there is a vertex $u\in V(X)$   such that
$\mu(u,t;Y)<k$ or not can be tested
 in $\mathrm{tm}(n,m,k)$ time and $\mathrm{sp}(n,m,k)$ space.  
\end{enumerate}
\end{lem}
\begin{proof} 
(i) The problem of computing $\mu(s,t;X)$ can be formulated
as a problem of finding a maximum flow in a weighted graph $(M,\omega_X)$ 
with an edge-capacity $\omega_X(e)$, $e\in E(M)$ and
a vertex-capacity $\omega_X(v)$, $v\in V(M)$, 
and $\mu(s,t;X)$  can be computed in $O(mn\log n)$ time and $O(n+m)$ space
by using the maximum flow algorithm \cite{AMO89,AMO93}.
Hence
$\mathrm{tm}(n,m,k)=O(mn\log n)$ and $\mathrm{sp}(n,m,k)=O(n+m)$.

(ii) Let  $t\in V(Y)\setminus V(X)$. 
 To find a vertex $u\in V(X)$ with $\mu(u,t;Y)<k$ if any by using (i) only once,
 we augment the weighted graph $(M,\omega_Y)$ into another weighted graph
  $(M^*,\omega_Y)$ 
 with a new vertex $s^*$
 and $|V(X)|$ new directed edges $e_u=(s^*,u)$, $u\in V(X)$
such that $\omega_Y(e_u):=k$.
We claim that $\mu(u,t;Y)\geq k$ for all vertices $u\in V(X)$ in $(M,\omega_Y)$ 
if and only if 
 $\mu(s^*,t;Y)\geq k$ in  $(M^*,\omega_Y)$. 

First consider the case of $\mu(s^*,t;Y)< k$ in $(M^*, \omega_Y)$; i.e., 
the graph $(M^*,\omega_Y)$ has an $s^*,t$-cut $C^*=(S,T)$ 
with $\omega_Y(\varepsilon(C^*))<k$, where $s^*\in S$ and $t\in T$.
Let $R=V(M^*)\setminus (S\cup T)$, where $R=V(M)\setminus (S\cup T)$.
Note that $X\subseteq S\cup R$, since otherwise $u\in T\cap V(X)$
 would mean that
$e_u=(s^*,u)\in E(S,T)$ and 
$\omega_Y(\varepsilon(C^*))\geq  \omega_Y(e_u)= k$,
contradicting that $\omega_Y(\varepsilon(C^*))<k$.
Also $S\cap V(X)\neq\emptyset$, 
since otherwise $V(X)\subseteq R$ would
mean that $\omega_Y(\varepsilon(C^*))\geq  \omega_Y(R)
\geq \omega_X(V(X))\geq k$,
contradicting that $\omega_Y(\varepsilon(C^*))<k$.
Let $u\in S\cap V(X)$. 
Then $C=(S\setminus\{s^*\},T)$ is a  $u,t$-cut in $(M,\omega_Y)$
with $\omega_Y(\varepsilon(C))\leq \omega_Y(\varepsilon(C^*))<k$.
This means that $\mu(u,t;Y)<k$. 

Next consider the case of $\mu(s^*,t;Y)\geq k$  in $(M^*, \omega_Y)$.
In this case, we show that  $\mu(u,t;Y)\geq k$ for all  vertices $u\in V(X)$. 
To derive a contradiction, assume that $\mu(u,t;Y)<k$
 for some vertex $u\in V(X)$; i.e.,
the graph $(M, \omega_Y)$ has a $u,t$-cut $C=(S,T)$ with $\omega_Y(\varepsilon(C))<k$.
Note that $T\cap V(X)=\emptyset$, since otherwise $u'\in  T\cap V(X)$ would 
contradict the assumption that $\mu(u,u';Y)\geq k$ holds 
 for all vertices $u,u'\in V(X)$.
Then $C'=(S'=S\cup \{s^*\},T)$ is an $s^*,t$-cut in $(M^*, \omega_Y)$,
and satisfies $\omega_Y(\varepsilon(C'))=\omega_Y(\varepsilon(C))<k$
since $T\cap V(X)=\emptyset$.
This, however, contradicts that $\mu(s^*,t;Y)\geq k$ holds  in $(M^*, \omega_Y)$. 

 By the claim, it suffices to test if $\mu(s^*,t;Y)\geq k$ or not  in 
$\mathrm{tm}(n,m,k)$ time and $\mathrm{sp}(n,m,k)$ space. 
\end{proof}

We denote by $\MC(M,w,\gamma,k)$ 
the family of $k$-connected sets
$X$ with $\omega_X(V(X))\ge k$
in a system $(M,w,\gamma)$. 
By  Lemmas~\ref{lem:transitive} and \ref{lem:monotone},
$\MC(M,w,\gamma,k)$ is transitive.
Let $\Lambda\subseteq V(M)\cup E(M)$ be a  subset.
Let   $\MC(M,w,\gamma,k,\Lambda)$ 
denote the family of components $X\in \MC(M,w,\gamma,k)$ 
such that $X\subseteq \Lambda$,
where we see that $\MC(M,w,\gamma,k,\Lambda)$ is also transitive.
We consider how to construct oracles $\mathrm{L}_1$ and
 $\mathrm{L}_2$ to the transitive system.
For two non-empty subsets $X\subseteq Y\subseteq \Lambda$,
let $\Max(Y)$ denote the family of maximal subsets
$Z\in \MC(M,w,\gamma,k,\Lambda)$
such that $Z\subseteq Y$, and
let $C_k(X;Y)$ denote
a maximal set $X^*\in \Max(Y)$ such that 
$X\subseteq X^*$; and
$C_k(X;Y)\triangleq \emptyset$
if no such set $X^*$ exists. 

\begin{lem}   \label{lem:maximal}   
For a monotone system  $(M,w,\gamma,\Lambda)$, 
let $\omega$ denote the corresponding  monotone meta-weight function.
Let $X\subseteq Y\subseteq  \Lambda$ be non-empty subsets  such that 
$\omega_X(V(X))\geq k$. 
Then 
\begin{enumerate}
\item[{\rm (i)}]  
 $X^*=C_k(X;Y)$  is uniquely determined;
\item[{\rm (ii)}] 
If there are vertices $u\in V(X)$ and $v\in V(Y)$ such that 
$\mu(u,v;Y)<k$, then $v\not\in V(X^*)$; 
\item[{\rm (iii)}]
Assume that  $\mu(u,v;Y)\geq k$ for all vertices $u\in V(X)$ 
and $v\in V(Y)\setminus V(X)$. 
Then $C_k(X;Y)=Y$ if  $\mu(u,u';Y)\geq k$ for all vertices $u,u'\in  V(X)$;
and  $C_k(X;Y)=\emptyset$ otherwise; and 
\item[{\rm (iv)}]
Finding $C_k(X;Y)$ can be done
 in $O(|Y|^2 \mathrm{tm}(n,m,k))$ time and $O(\mathrm{sp}(n,m,k) +|Y|)$ space.  
\end{enumerate}
\end{lem}
\begin{proof} 
(i) To derive a contradiction, assume that 
 there are two maximal sets $X_1,X_2\in \Max(Y)$
 such that $X\subseteq X_1\cap X_2$. 
From this and the monotonicity of  $\omega$,
it holds that  
$\omega_{X_1\cup X_2}(V(X_1\cup X_2))\geq$
$\omega_{X_1\cap X_2}(V(X_1\cap X_2))$
$\geq \omega_X(V(X))\geq k$.
From this and Lemma~\ref{lem:connectivity}, 
 $X_1\cup X_2$ is also $k$-connected and  
   $X_1\cup X_2\in \Max(Y)$, contradicting
   the maximality of $X_1$ and $X_2$. 
   Therefore $C_k(X;Y)$ is unique.
   
 (ii) When $C_k(X;Y)=\emptyset$, $v\not\in C_k(X;Y)$ is trivial.
 Assume that $C_k(X;Y)=X^*\in \Max(Y)$. 
By the monotonicity of  $\omega$  and $X^*\subseteq Y$,  
it holds that  $\mu(u,v;X^*)\leq  \mu(u,v;Y)  <k$.
Hence $u,v\in V(X^*)$ would contradict the $k$-connectivity of $X^*$. 
Since $u\in  V(X^*)$, we have $v\not\in V(X^*)$. 

(iii)  Obviously if $\mu(u,u';Y)< k$ for some vertices $u,u'\in V(X)$,
then  no subset $Y'$ of $Y$ with $X\subseteq Y'$  can be $k$-connected,
and $C_k(X;Y)=\emptyset$. 
Assume that $\mu(u,u';Y)\geq k$ for all vertices $u,u'\in V(X)$.  
By the monotonicity of  $\omega$  and $X\subseteq Y$,  
it holds that  $\omega_{Y}(V(Y))\geq \omega_X(V(X))\geq k$.
To prove that $C_k(X;Y)=Y$, it suffices to show 
that $\mu(u,v;Y)\geq k$ for all pairs of vertices $u,v\in V(Y)$.
  By assumption, $\mu(u,v;Y)\geq k$ for all vertices $u\in V(X)$ and $v\in V(Y)$.
  To derive a contradiction, assume that there is a pair of vertices
   $s,t\in V(Y)\setminus V(X)$
  with $\mu(s,t;Y)< k$; i.e.,
  there is an $s,t$-cut $C=(S,T)$ with $\omega_Y(\varepsilon(C))<k$.
  Let $R=V(M)\setminus S\cup T$. 
We observe that $V(X)\subseteq R$, since
  $u\in V(X)\cap S$ (resp.,  $u\in V(X)\cap T$) would imply
  that $C$ is a $u,t$-cut (resp., $s,u$-cut), contradicting that 
  $\mu(u,v;Y)\geq k$ for all vertices $v\in V(Y)\setminus V(X)$. 
By the monotonicity of $\omega$ and 
$V(X)\subseteq R$, it would hold that
   $k\leq \omega_X(V(X))\leq \omega_Y(R)\leq \omega_Y(\varepsilon(C))<k$,
a contradiction. 
 
(iv) We can find $C_k(X;Y)$ as follows.
Based on (ii), we first remove the set $Z_V$ of
 all vertices $t\in V(M)\cap(Y\setminus X)$ such that $\mu(u,t;Y)<k$ 
for some vertex $u\in V(X)$ 
and the set $Z_E$ of 
 all edges $e\in E(M)\cap(Y\setminus X)$ such that $\mu(u,t;Y)<k$ 
for some vertices $u\in V(X)$ and $t\in V(\{e\})$
so that $C_k(X;Y)=C_k(X;Y')$ holds for $Y'=Y\setminus (Z_V\cup Z_E)$.
For a fixed vertex $t\in  V(M)\cap(Y\setminus X)$ 
or $t\in V(\{e\})$ with an edge $e\in E(M)\cap(Y\setminus X)$, 
we can test if there is a  vertex $u\in X$
such that $\mu(u,t;Y)<k$ or not in $O(\mathrm{tm}(n,m,k))$ time 
and $O(\mathrm{sp}(n,m,k))$ space
by Lemma~\ref{lem:flow}(ii).
Hence finding such a set $Z_V\cup Z_E$
 takes $O(|Y\setminus X|\mathrm{tm}(n,m,k))$ time 
and $O(\mathrm{sp}(n,m,k)+|Z_V\cup Z_E|)$ space.
We repeat the above procedure until there is no pair 
of vertices $u\in V(X)$ and $v\in V(Y')\setminus V(X)$
after executing at most $|Y\setminus X|$ repetitions taking 
$O(|Y\setminus X|^2\mathrm{tm}(n,m,k))$ time and
 $O(\mathrm{sp}(n,m,k)+|Y\setminus X|)$ space.

\begin{algorithm}[h]
  \caption{{\sc Maximal}$(X;Y)$: 
  Finding the maximal set in  $\MC(M,w,\gamma,k,\Lambda)$
   that contains   a specified set  }
  \label{alg:maximal_specified}
  \begin{algorithmic}[1]
\Require A monotone system $(M,w,\gamma)$, a real $k\geq 0$,
a subset $\Lambda\subseteq V(M)\cup E(M)$ and non-empty 
subsets $X\subseteq Y\subseteq \Lambda$ 
such that $\omega_X(V(X))\geq k$.
\Ensure    $C_k(X;Y)$
 \State   $Y':=Y$; 
\While {there are vertices $u\in V(X)$ and $t\in V(Y')\setminus V(X)$ 
such that 
$\mu(u,t;Y')<k$}
\State   $Z_V:= \{ t\in V(M)\cap(Y'\setminus X) \mid \mbox{$\mu(u,t;Y')<k$
  for some vertex $u\in V(X)$}\}$;
\State   $Z_E:= \{ e\in E(M)\cap(Y'\setminus X) \mid   \mu(u,t;Y')<k$  \par
    \hskip\algorithmicindent  for some vertices $u\in V(X)$ and $t\in V(\{e\})\}$;
\State $Y':=Y'\setminus (Z_V\cup Z_E)$
\EndWhile;
\If {$\mu(u,u';Y')\geq k$ for all vertices $u,u'\in V(X)$}
  \State Output $Y'$ as $C_k(X;Y)$
\Else
  \State Output $\emptyset$ as $C_k(X;Y)$
\EndIf
  \end{algorithmic}\end{algorithm}

Based on (iii), we finally conclude that $C_k(X;Y)=Y'$ ($C_k(X;Y)=\emptyset$) 
 if there is not pair of  vertices $u,u'\in V(X)$ such that $\mu(u,u';Y')<k$ 
 (resp., otherwise), which takes  
$O(|X|^2\mathrm{tm}(n,m,k))$ time and $O(\mathrm{sp}(n,m,k))$ space
 by Lemma~\ref{lem:flow}(i).

An entire algorithm is described in Algorithm~\ref{alg:maximal_specified}.
The time and space complexities are then 
$O(|Y|^2\mathrm{tm}(n,m,k))$ time and $O(\mathrm{sp}(n,m,k) +|Y|)$, respectively. 
\end{proof}

 By the lemma, oracle $\mathrm{L}_1(X;Y)$  to a monotone system
 $(M,w,\gamma)$ runs 
   in $\theta_{\mathrm{1,t}}=O(|Y|^2\mathrm{tm}(n,m,k))$ time and  
$\theta_{\mathrm{1,s}}=O(\mathrm{sp}(n,m,k) +|Y|)$ space.

\bigskip
\noindent
{\bf $k$-core} 
For a   system $(M,w,\gamma,\Lambda)$,  
we define a {\em $k$-core} of a subset $Y\subseteq  \Lambda$ 
to be a  subset $Z$ of $Y$
such that   $\omega_Z(V(Z))\geq k$ and any proper subset $Z'$ of $Z$
satisfies  $\omega_{Z'}(V(Z'))< k$.

\begin{lem}   \label{lem:core}   
Let $(M,w,\gamma,\Lambda)$ be a monotone system, and $Y$
be a subset of $\Lambda$.  
For   the family $\mathcal{K}$ of all $k$-cores of $Y$,
it holds that  $\Max(Y)= \bigcup_{Z\in \mathcal{K}} \{C_k(Z;Y)\}$
and $|\Max(Y)|\leq |\mathcal{K}|$. 
Given $\mathcal{K}$, $\Max(Y)$ can be obtained
in $O(|\mathcal{K}|( |Y|^2 \mathrm{tm}(n,m,k)+|Y|\log |\mathcal{K}|))$ time and 
$O(\mathrm{sp}(n,m,k)+|\mathcal{K}|\cdot|Y|)$ space. 
\end{lem}    
\begin{proof}
Clearly each set $X\in  \Max(Y)$ satisfies $\omega_X(V(X))\geq k$ and
contains a $k$-core $Z\in \mathcal{K}$, where 
$C_k(Z;Y)\neq\emptyset$ and $C_k(Z;Y)=X$ holds by the uniqueness in 
Lemma~\ref{lem:maximal}(i).
Therefore $\Max(Y)=\bigcup_{Z\in \mathcal{K}} \{C_k(Z;Y)\}$,
from which  $|\Max(Y)|\leq |\mathcal{K}|$ follows. 
 Given $\mathcal{K}$, we compute $C_k(Z;Y)$ for each set   $Z\in \mathcal{K}$
 taking  $O(|Y|^2 \mathrm{tm}(n,m,k))$ time and 
$O(\mathrm{sp}(n,m,k)+|Y|)$ space 
 by Lemma~\ref{lem:maximal}(iv).
 We can test if the same set $X\in  \Max(Y)$ has been generated or not
 in $O(|Y|\log |\mathcal{K}|)$ time and $O(|\mathcal{K}|\cdot|Y|)$ space. 
 Therefore $\mathcal{X}$ can be constructed  
in $O(|\mathcal{K}|( |Y|^2 \mathrm{tm}(n,m,k)+|Y|\log |\mathcal{K}|))$ time and 
$O(\mathrm{sp}(n,m,k)+|\mathcal{K}|\cdot|Y|)$ space.  
\end{proof}

By the lemma, oracle $\mathrm{L}_2(Y)$  to a monotone system 
$(M,w,\gamma,\Lambda)$ runs 
in $\theta_{\mathrm{2,t}}=O(|\mathcal{K}|( |Y|^2 \mathrm{tm}(n,m,k)
+|Y|\log |\mathcal{K}|))$ time and  
$\theta_{\mathrm{2,s}}=O(\mathrm{sp}(n,m,k)+|\mathcal{K}|\cdot|Y|)$ space,
where we assume that the family $\mathcal{K}$ of $k$-cores of 
$Y$ is given as input.

\section{Connector Enumeration Problem}
\label{sec:app.conn}

This section treats   the   case 
where $G$ is an undirected graph 
and $\MC$ is the family of all vertex subsets that induce connected subgraphs. 
 
\subsection{Problem Description}
Assume that we are given a tuple $(G,I,\sigma)$ with
an undirected graph $G$,
a set $I$ of items,
and a function $\sigma:V(G)\to 2^I$.
For a subset $X\subseteq V(G)$, 
let  $\Item(X)$ denote the common item set $\bigcap_{u\in X}\sigma(u)$.
A subset $X\subseteq V(G)$ such that $G[X]$ is connected
called a {\em connector},
if  for any vertex $v\in V(G)\setminus X$,
$G[X\cup\{v\}]$ is not connected or 
$\Item(X\cup\{v\})\subsetneq \Item(X)$; i.e., 
there is  no proper superset $Y$ of $X$ 
such that $G[Y]$ is connected and $\Item(Y)=\Item(X)$. 

The problem of enumerating all connectors
is called the {\em connector enumeration problem}
 in the literature~\cite{HMSN2.2018,HMSN.2019,O.2017,SSF.2010},
which has applications in biology. 

\figref{instance} illustrates  a brief example of an 
instance of the connector enumeration problem.

We show that \thmref{main} yields the first polynomial-delay algorithm
 for the connector enumeration problem.

\begin{figure}[t!]
  \centering
  \begin{tikzpicture}
    \node [circle,draw,fill=white] (v1) at (0,2) [label=left:{1,2,3}] {$v_1$};
    \node [circle,draw,fill=white] (v2) at (2,2) [label=right:{1,3}] {$v_2$};
    \node [circle,draw,fill=white] (v3) at (0,0) [label=left:{1,2}] {$v_3$};
    \node [circle,draw,fill=white] (v4) at (2,0) [label=right:3] {$v_4$};
    \draw (v1) -- (v2);
    \draw (v1) -- (v3);
    \draw (v2) -- (v3);
    \draw (v3) -- (v4);
  \end{tikzpicture}
  \caption{An instance of the connector enumeration problem: 
  it has connectors $\{v_1\}$, $\{v_4\}$, $\{v_1,v_2\}$, $\{v_1,v_3\}$,
   $\{v_1,v_2,v_3\}$,
    and $\{v_1,v_2,v_3,v_4\}$, where an item is represented by an integer}
  \label{fig:instance}
\end{figure}
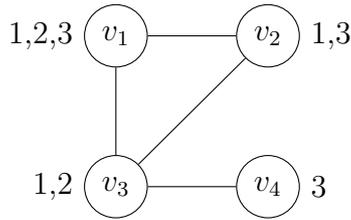

\subsection{Background}
\paragraph{Application.}
A graph with an item set, or an attributed graph,
is useful to represent many existing networks
such as social networks and biological networks. 
Some papers in the literature
have reported applications of the connector enumeration problem in biology. 
Seki and Sese~\cite{SS.2008} considered
a biological network such that a vertex corresponds to a gene
and an edge represents a protein-protein interaction between genes. 
A gene produces RNAs under a certain condition, and the phenomenon 
is called gene expression.
A condition at which gene expression occurs is given to a vertex as an item.
A biologist is particularly interested in a large-sized connector 
with a large common item set,
that is, a large connected set of genes that make expressions simultaneously
under common (possibly complex) conditions.

More recently, Alokshiya et al.~\cite{ASA.2019}
proposed a new algorithm for enumerating all connected induced subgraphs (CISs)
of a given (non-attributed) graph.
They applied the algorithm to find
{\em maximal cohesive patterns} in
BIOGRID protein-protein interaction network~\cite{BIOGRID.2015},
where a maximal cohesive pattern is defined as 
a connector that is maximal among those $X$ satisfying $|\Item(X)|\ge\theta$
for a threshold $\theta$. 


\paragraph{Related studies.}
The connector enumeration problem is a generalization of
the frequent item set mining problem~\cite{AIS.1993}, 
a well-known problem in data mining, such that $G$ is
a clique and a vertex corresponds to a transaction.

For an attributed graph, community detection~\cite{LSHZ.2018} and
frequent subgraph mining~\cite{IWM.2000} are among significant
 graph mining problems.
The latter asks to enumerate all subgraphs
that appear in a given set of attributed graphs ``frequently,''
where the graph isomorphism is defined by taking into account the items.
For the problem, 
gSpan~\cite{YH.2002} should be one of the most successful algorithms. 
The algorithm enumerates all frequent subgraphs
by growing up a search tree.
In the search tree,
a node in a depth $d$ corresponds to a subgraph
that consists of $d$ vertices, and
a node $u$ is the parent of a node $v$
if the subgraph for $v$ is obtained by adding one vertex to
the subgraph for $u$.

For the connector enumeration problem, 
Sese et al.~\cite{SSF.2010} proposed the first algorithm, named COPINE,
which explores the search space by utilizing
the similar search tree as gSpan. 
Okuno et al.~\cite{OHNYS.2014,OHNYS.2016} and Okuno~\cite{O.2017}
studied parallelization of COPINE. 
No algorithm with a theoretical time bound had been known 
until Haraguchi et al.~\cite{HMSN2.2018,HMSN.2019} proposed 
an output-polynomial algorithm, named COOMA, based on 
  a dynamic programming method. 

\subsection{Formulation by Transitive System}
Let us consider formulating
the connector enumeration problem
by means of a transitive system.
For a given instance $(G,I,\sigma)$ of the connector enumeration problem, 
let $\MC_G$ denote the family of
all vertex subsets $X$ such that
the induced subgraph $G[X]$ is connected,
where we regard $G[X]$ with $|X|=1$ (resp., $X=\emptyset$)
as connected (resp., disconnected). 
We see that $(V(G),\MC_G)$ is a transitive system
since, for any $X,Y\in\MC_G$,
$G[X\cup Y]$ is connected whenever $G[X\cap Y]$ is connected. 

Let $n=|V(G)|$ and $m=|E(G)|$. 
We can implement the oracles L$_1$ and L$_2$
so that they run in $O(n+m)$ time and space
(i.e., $\theta_{i,\mathrm{t}}=O(n+m)$, $i=1,2$, and 
$\theta_{i,\mathrm{s}}=O(n+m)$, $i=1,2$)
since they are realized by
conventional graph search (e.g., DFS or BFS). 
We can take the upper bound $\delta(Y)=|Y|$, 
which exactly satisfies our assumption that
$\delta(X)\le\delta(Y)$ holds for subsets $X\subseteq Y\subseteq V$. 

For any $X\subseteq V$,
the followings are equivalent:
\begin{itemize}
  \item $X$ is a connector
    for $(G,I,\sigma)$; and
  \item $X$ is a solution for the instance $(V(G),\MC_G,I,\sigma)$.
\end{itemize}
In \figref{hasse},
we show the Hasse diagram
of the transitive system $(V(G),\MC_G)$
for the instance $(G,I,\sigma)$ in \figref{instance},
along with the solutions for $(V(G),\MC_G,I,\sigma)$.

\begin{figure}[t!]
  \centering
  \begin{tikzpicture}
    \node [rectangle,draw,fill=lightgray,very thick] (S1234) at (3,4.5) [label=above:{$\emptyset$}] {$v_1,v_2,v_3,v_4$};
    \node [rectangle,draw,fill=lightgray,very thick] (S123) at (1,3) [label=above:{1}] {$v_1,v_2,v_3$};
    \node [rectangle,draw,fill=white] (S134) at (3,3) [label=above left:{$\emptyset$}] {$v_1,v_3,v_4$};
    \node [rectangle,draw,fill=white] (S234) at (5,3) [label=above:{$\emptyset$}] {$v_2,v_3,v_4$};
    \node [rectangle,draw,fill=lightgray,very thick] (S12) at (0,1.5) [label=left:{1,3}]{$v_1,v_2$};
    \node [rectangle,draw,fill=lightgray,very thick] (S13) at (2,1.5) [label=left:{1,2}] {$v_1,v_3$};
    \node [rectangle,draw,fill=white] (S23) at (4,1.5) [label=right:{1}] {$v_2,v_3$};
    \node [rectangle,draw,fill=white] (S34) at (6,1.5) [label=right:{$\emptyset$}] {$v_3,v_4$};
    \node [rectangle,draw,fill=lightgray,very thick] (S1) at (0,0) [label=left:{1,2,3}] {$v_1$};
    \node [rectangle,draw,fill=white] (S2) at (2,0) [label=left:{1,3}] {$v_2$};
    \node [rectangle,draw,fill=white] (S3) at (4,0) [label=right:{1,2}] {$v_3$};
    \node [rectangle,draw,fill=lightgray,very thick] (S4) at (6,0) [label=right:{3}] {$v_4$};
    \draw (S1234) -- (S123);
    \draw (S1234) -- (S134);
    \draw (S1234) -- (S234);
    \draw (S123) -- (S12);
    \draw (S123) -- (S13);
    \draw (S123) -- (S23);
    \draw (S134) -- (S13);
    \draw (S134) -- (S34);
    \draw (S234) -- (S23);
    \draw (S234) -- (S34);
    \draw (S12) -- (S1);
    \draw (S13) -- (S1);
    \draw (S23) -- (S2);
    \draw (S34) -- (S3);
    \draw (S12) -- (S2);
    \draw (S13) -- (S3);
    \draw (S23) -- (S3);
    \draw (S34) -- (S4);
  \end{tikzpicture}
  \caption{Hasse diagram of the transitive system $(V,\MC_G)$
    of the instance $(V(G),\MC_G,I,\sigma)$ from \figref{instance},
    where common item sets are indicated by integers and solutions 
    are indicated by shade}
  \label{fig:hasse}
\end{figure}
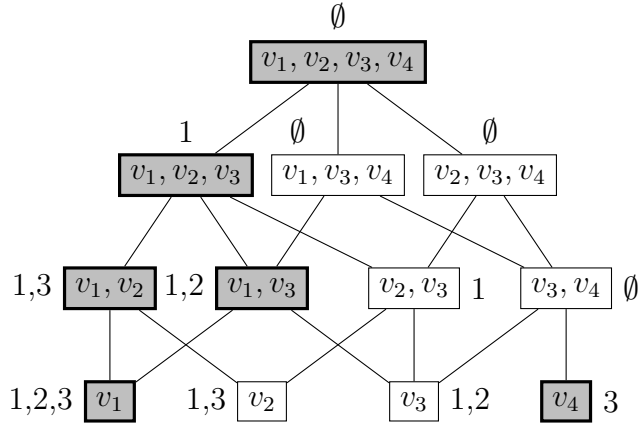
 
\begin{thm}   \label{thm:connect}
  Given an instance $(G,I,\sigma)$
  of the connector enumeration problem in a graph $G$,
   all connectors can be enumerated
  in $O(q^2(n+m)n)$ 
  delay and   $O((q+n+m)n)$ space,
  where $n=|V(G)|$, $m=|E(G)|$ and $q=|I|$. 
\end{thm}
\begin{proof}
  The connector enumeration problem for $(G,I,\sigma)$
  is solved by enumerating all solutions for the instance $(V(G),\MC_G,I,\sigma)$.
  For the transitive system $(V(G),\MC_G)$,
  we see that $\theta_{i,\mathrm{t}}=O(n+m)$, $i=1,2$, 
  $\theta_{i,\mathrm{s}}=O(n+m)$, $i=1,2$, and  
  $\delta(Y)=O(|Y|)=O(n)$.
  By \thmref{main}, we can enumerate all solutions in $\MS$
  in $O(q^2(n+m)n)$ 
  delay and in $O((q+n+m)n)$ space. 
\end{proof}

\subsection{Enumerating Connectors under Various Connectivity Conditions}
\label{sec:app.ext}

In addition to the system $(V(G),\MC_G)$,
we may obtain an alternative  
transitive system by selecting a different notion of connectivity such as
the edge- or vertex-connectivity on a digraph or an undirected graph.
This section presents  two examples of transitive systems
based on high graph connectivity
using the  result in Section~\ref{sec:app.ext.mixed}.

\paragraph{Edge- and Vertex-Connectivity in Mixed Graph} 
Let $G$ be a mixed graph with $n$ vertices and $m$ edges.  
We define a {\em path}  from a vertex $u$ to a vertex $v$ 
(or a {\em $u,v$-path})
 in $G$   to be a subgraph $P$ of $G$ such that
$V(P)=\{v_1~(=u),v_2,\ldots,v_p~(=v)\}$,
$E(P)=\{e_1,e_2,\ldots,e_{p-1}\}$ 
and $e_i=v_iv_{i+1}\in \overline{E}(P)$ or
 $e_i=(v_i,v_{i+1})\in \vec{E}(P)$. 
Let $s,t\in V(G)$ be two vertices in $G$. 
Let $\lambda(s,t;G)$ denote the minimum size $|F|$ of 
a subset $F\subseteq E(G)$ so that the graph $G-F$ 
obtained from $G$ by removing edges in $F$ 
has no $s,t$-path.
Let $\kappa(s,t;G)$ denote the minimum size $|S|$ of
a subset $S\subseteq E(G)\cup (V(G)\setminus\{s,t\})$ 
 to be removed from $G$ so that the graph $G-S$ 
obtained from $G$ by removing vertices and edges in $S$ 
has no $s,t$-path,
where such a minimum subset $S$ can be chosen so that 
$S\setminus E(s,t)\subseteq V(G)$.
By Menger's theorem~\cite{Me27}, 
$\lambda(s,t;G)$ (resp., $\kappa(s,t;G)$) is equal to 
the maximum number of edge-disjoint (resp., internally disjoint) $s,t$-paths.
We can test whether 
$\lambda(s,t;G)\geq k$ (resp., $\kappa(s,t;G)\geq k$)
or not  in $O(\min\{k,n\}m)$
 (resp., $O(\min\{k,n^{1/2}\}m)$) time \cite{AMO89,AMO93}.
A graph $G$ is called {\em $k$-edge-connected} if 
$|V(G)|\geq 1$ and 
 $\lambda(u,v;G)\geq k$ 
 for any two vertices $u,v\in V(G)$. 
A graph $G$ is called {\em $k$-vertex-connected} if 
$|V(G)|\geq k+1$ and  $\kappa(u,v;G)\geq k$ 
 for any two vertices $u,v\in V(G)$. 
In the following, we show two examples of transitive systems
based on graph connectivity. 

\subsection{Vertex Subsets Highly-connected over the Entire Graph}

Given a mixed graph $G$, we define ``$k$-connected set'' based
on the connectivity of the entire graph $G$. 
Let us call a subset $X\subseteq V(G)$ {\em $k$-edge-connected}
if $|X|=1$ or for any two vertices $u,v\in X$, 
 $\lambda(u,v;G)\geq k$.
 Let $\MC_{k,\mathrm{edge}}$ denote the family
 of $k$-edge-connected sets in $G$. 
Let us call a subset $X\subseteq V(G)$ {\em $k$-vertex-connected}
if $|X|\geq k$ or for any two vertices $u,v\in X$, 
 $\kappa(u,v;G)\geq k$.
 Let $\MC_{k,\mathrm{vertex}}$ denote the family
 of $k$-vertex-connected sets in $G$.

\begin{lem}   \label{lem:construct_oracle_entire}   
Let $G$ be a mixed graph   and $k\geq 0$ be an integer,
  where $n=|V(G)|$ and $m=|E(G)|$. 
\begin{enumerate}
\item[{\rm (i)}] 
The family $\MC=\MC_{k,\mathrm{edge}}$   is   transitive.
For any non-empty subsets $X\subseteq Y\subseteq V(G)$,
it holds $|\Max(Y)|\leq |Y|$,
and
oracles 
$\mathrm{L}_1(X;Y)$ and $\mathrm{L}_2(Y)$ run 
in   $O(n^2)$  time and space
after an $O(n^2\min\{k,n\}m)$-time and $O(n^2)$-space preprocessing; and 
\item[{\rm (ii)}] 
The family  $\MC=\MC_{k,\mathrm{vertex}}$ is   transitive.
For any non-empty subsets $X\subseteq Y\subseteq V(G)$,  
it holds $|\Max(Y)|\leq {{|Y|}\choose{k}}$, 
oracle  $\mathrm{L}_1(X;Y)$ runs in $O(n^2)$ time and $O(n^2)$ space, and 
oracle $\mathrm{L}_2(Y)$ runs  in  $O(|Y|^{k}n^2)$ time and $O(|Y|^{k}n)$ space,  
after an $O(n^2\min\{k,n^{1/2}\}m)$-time and $O(n^2)$-space preprocessing.
\end{enumerate}
\end{lem}
\begin{proof}  
Let
  $(M,w,\gamma,k,\Lambda)$  be  a  system such that 
 a mixed graph $M:=G$,     $\Lambda:=V(G)$, and 
 a weight function $w$ and
  a coefficient function 
  $\gamma=(\alpha,\overline{\alpha},\alpha^+,\alpha^-,\beta)$
  such that 
  $\alpha(e):=\overline{\alpha}(e):=\alpha^+(e):=\alpha^-(e):=1$
  for each edge $e\in E(G)$,    and 
  $\beta(a):=1$ for each element $a\in V(G)\cup E(G)$,
  where we see that   $\gamma$ is monotone and 
  the family $\MC(M,w,\gamma,k,\Lambda)$ is transitive
  by Lemmas~\ref{lem:transitive} and \ref{lem:monotone}. 
  
  (i)  We set weight $w$ so that $w(e):=1$  for each edge $e\in E(G)$
   and $w(v):=k$  for each vertex $v\in V(G)$.
We claim that 
 $\MC_{k,\mathrm{edge}}$ is equal to $\MC(M,w,\gamma,k,\Lambda)$,
 where the latter is the family of non-empty subsets  
  $X\subseteq \Lambda$ with 
 $\omega_X(V(X))\geq k$ such that
  $|V(X)|=1$ or $\mu(u,v;X)\geq k$ for each pair of vertices $u,v\in V(X)$.  
 Note that $\omega_X(V(X))\geq w(V(X))=k|X|$ for any non-empty set
 $X\subseteq \Lambda$.
 Then every set $X\subseteq V(G)$ with $|X|=1$ belongs
 to both 
 $\MC_{k,\mathrm{edge}}$ and $\MC(M,w,\gamma,k,\Lambda)$.
Let $X$ be a subset of $V(G)$ with  $|V(X)|=|X|\geq 2$. 
By definition of coefficient function $\gamma$ and weight $w$ in $G$,
we see that 
 $\mu(u,v;X)=\lambda(u,v;G)$ holds 
for any two vertices $u,v\in V(X)$.  
 This means that 
 $\MC_{k,\mathrm{edge}}=\MC(M,w,\gamma,k,\Lambda)$,
 proving the claim. 
 
  We define the  {\em auxiliary graph} 
 $G_{k,\mathrm{edge}}^*$    to be an undirected graph  
 such that\\
~~ $V(G_{k,\mathrm{edge}}^*)=V(G)$,\\
~~ $E(G_{k,\mathrm{edge}}^*)
=\{uv\mid u,v\in V(G)\mbox{ such that $\lambda(u,v;G)\geq k$ 
and $\lambda(v,u;G)\geq k$} \}$.\\
 We can construct $G_{k,\mathrm{edge}}^*$     in $O(n^2\min\{k,n\}m)$
  time and $O(n^2)$ space. 
Observe that a non-empty  subset
  $X\subseteq V(G)$ belongs to $\MC_{k,\mathrm{edge}}$ 
 if and only if $w(X)\geq k$ and $X$ forms a clique in $G_{k,\mathrm{edge}}^*$.
 For edge-connectivity, we easily see that 
 $\lambda(x,y;G), \lambda(y,x;G),$ $\lambda(y,z;G), \lambda(z,y;G)\geq k$ 
 imply $\lambda(x,z;G), \lambda(z,x;G)\geq k$.
 Hence $G_{k,\mathrm{edge}}^*$ is a disjoint union of cliques,
 and for $\MC=\MC_{k,\mathrm{edge}}$,  the family $\Max(Y)$  
  is also a disjoint union of cliques in the induced subgraph 
  $G_{k,\mathrm{edge}}^*[Y]$.
  This means that  
  $|\Max(Y)|\leq |Y|$ holds and $\Max(Y)$
  is
 found in $O(n^2)$ time as the set of connected components in
   $G_{k,\mathrm{edge}}^*$.
 For $\MC=\MC_{k,\mathrm{edge}}$, 
 $\mathrm{L}_1(X;Y)$ and $\mathrm{L}_2(Y)$ run 
   in   $O(n^2)$  time and space
   after an $O(n^2\min\{k,n\}m)$-time and $O(n^2)$-space preprocessing.
 
 (ii)  
We set  weight $w$ so that $w(e):=1$ for each edge $e\in E(G)$
 and $w(v):=1$ for each vertex $v\in V(G)$.  
We claim that   $\MC_{k,\mathrm{vertex}}$ is equal to  
$\MC(M,w,\gamma,k,\Lambda)$. 
Note that $\omega_X(V(X))= w(V(X))=|X|$
for any non-empty set  $X\subseteq \Lambda$.
Let $X$ be a subset of $V(G)$ with  $|V(X)|=|X|<k$. 
Then $X$ is not $k$-vertex-connected in $G$ and
$X$ is not $k$-connected in the system $(M,w,\gamma,k,\Lambda)$. 
Let $X$ be a subset of $V(G)$ with  $|V(X)|=|X|\geq k$. 
By definition of coefficient function $\gamma$ and weight $w$ in $G$,
we see that 
 $\mu(u,v;X)=\kappa(u,v;G)$ holds 
for any two vertices $u,v\in V(X)$.  
 This means that 
 $\MC_{k,\mathrm{vertex}}=\MC(M,w,\gamma,k,\Lambda)$,
 proving the claim. 

 We define the  {\em auxiliary graph} 
  $G_{k,\mathrm{vertex}}^*$  to be an undirected graph 
 such that\\
~~ $V(G_{k,\mathrm{vertex}}^*)=V(G)$,\\
~~ $E(G_{k,\mathrm{vertex}}^*)
=\{uv\mid u,v\in V(G)\mbox{ such that $\kappa(u,v;G)\geq k$
 and $\kappa(v,u;G)\geq k$} \}$.\\
We can construct $G_{k,\mathrm{vertex}}^*$ in 
 $O(n^2\min\{k,n^{1/2}\}m)$ time and $O(n^2)$ space. 
Observe that a non-empty  subset  $X\subseteq V(G)$ 
belongs to $\MC_{k,\mathrm{vertex}}$ 
 if and only if $w(X)\geq k$ and $X$ forms a clique in $G_{k,\mathrm{vertex}}^*$.
 
Let $\MC=\MC_{k,\mathrm{vertex}}$.
For   subsets $X\subseteq Y\subseteq V(G)$ such that $|X|\geq k$,
 a   maximal set $Z \in\Max(Y)$ with $X\subseteq Z$
 is the unique set $C_k(X;Y)$ by Lemma~\ref{lem:maximal}.
 Hence  $C_k(X;Y)$ can be found 
 in $O(n^2)$ time and  space  by constructing
  the unique maximal clique containing $X$
 in the induced subgraph $G_{k,\mathrm{vertex}}^*[Y]$. 
 Let $\mathcal{K}$ be the family of $k$-cores; i.e., 
 subsets of exactly $k$ vertices in $Y$,
 which can be constructed in $O(|Y|^k)$ time. 
 By Lemma~\ref{lem:core}, $|\Max(Y)|\leq |\mathcal{K}|={{|Y|}\choose{k}}$ holds,
 and  we can construct  $\Max(Y)$
 by computing $C_k(Z;Y)$ for all sets $Z\in \mathcal{K}$, taking 
  $O(|Y|^{k}n^2)$ time and $O(|Y|^{k}n)$ space.
\end{proof}

Using \thmref{main} and \lemref{construct_oracle_entire}, 
we have the following theorem
on the time delay and the space complexity
of enumeration of connectors that
are $k$-edge-connected or $k$-vertex-connected.

\begin{thm}
  \label{thm:entire}
  Let $(G,I,\sigma)$ be an instance on a mixed graph $G$ 
  and $k\ge0$ be an integer,
  where   $n=|V(G)|$, $m=|E(G)|$, and $q=|I|$. 
\begin{enumerate}
\item[{\rm (i)}] 
All  $k$-edge-connected connectors can be enumerated 
  in $O(q^2n^3)$  delay and $O(qn+n^3)$ space,
  after an $O(n^2\min\{k,n\}m)$-time and $O(n^2)$-space preprocessing.
\item[{\rm (ii)}] 
All  $k$-vertex-connected connectors can be enumerated 
  in $O(q^2n^{k+2})$ delay and  $O(qn+n^{k+2})$ space,
  after an $O(n^2\min\{k,n^{1/2}\}m)$-time
   and $O(n^2)$-space preprocessing.
\end{enumerate}
\end{thm}
\begin{proof}
  Recall that, for $Y\subseteq V$,
  $\delta(Y)$ denotes an upper bound on $|\MC_{\max}(Y)|$.
  In both (i) and (ii), $\theta_{\rho,\mathrm{t}}$ and
   $\theta_{\rho, \mathrm{s}}$ can be regarded 
   as $O(1)$ since the volume function is not used 
   anywhere in this context. 
  
  (i) By \lemref{construct_oracle_entire}(i), 
  we have $\theta_{1,\mathrm{t}}=\theta_{2,\mathrm{t}}=O(n^2)$
  and $\theta_{1,\mathrm{s}}=\theta_{2,\mathrm{s}}=O(n^2)$,
  and we can set $\delta(Y)=n$ for any $Y\subseteq V$.
  By \thmref{main}, we have the time delay 
  $O\big(q\theta_{2,\mathrm{t}} 
  + (q(n+\theta_{1,\mathrm{t}})+\theta_{\rho,\mathrm{t}})q 
  \delta(V)\big)
  =O(q^2n^3)$
  and the space complexity 
  $O\big((q+n+\theta_{1,\mathrm{s}}+\theta_{2,\mathrm{s}}
  + \theta_{\rho,\mathrm{s}}) n\big)
  =O(qn+n^3)$. 
  
  (ii) By \lemref{construct_oracle_entire}(ii), 
  we have $\theta_{1,\mathrm{t}}=O(n^2)$, 
  $\theta_{2,\mathrm{t}}=O(n^{k+2})$,
  $\theta_{1,\mathrm{s}}=O(n^2)$, and 
  $\theta_{2,\mathrm{s}}=O(n^{k+1})$,
  and we can set $\delta(Y)=n^k$ for any $Y\subseteq V$.
  By \thmref{main}, we have the time delay 
  $O\big(q\theta_{2,\mathrm{t}} 
  +(q(n+\theta_{1,\mathrm{t}})+\theta_{\rho,\mathrm{t}})q
  \delta(V)\big)
  =O(q^2n^{k+2})$
  and the space complexity 
  $O\big((q+n+\theta_{1,\mathrm{s}}+\theta_{2,\mathrm{s}}
  + \theta_{\rho,\mathrm{s}}) n\big)
  =O(qn+n^{k+2})$. 

  For preprocessing, 
  the time and space complexities are immediate from \lemref{construct_oracle_entire}
  both for (i) and (ii).
\end{proof}

\subsection{Highly-connected Induced Subgraphs}

Given a mixed graph $G$, we define a ``$k$-connected set'' $X$ based
on the connectivity of the induced graph $G[X]$. 
Define  $\MC_{k,\mathrm{edge}}^{\mathrm{in}}$ to be the family
(resp.,  $\MC_{k,\mathrm{vertex}}^{\mathrm{in}}$) 
of subsets $X\in V(G)$ such that 
the induced graph $G[X]$ is $k$-edge-connected (resp., $k$-vertex-connected).

\begin{lem}   \label{lem:construct_oracle_inner}
Let $G$ be a mixed graph   and $k\geq 0$ be an integer,
  where $n=|V(G)|$ and $m=|E(G)|$. 
\begin{enumerate}
\item[{\rm (i)}] 
The family $\MC=\MC_{k,\mathrm{edge}}^{\mathrm{in}}$ is transitive.  
For any non-empty subsets $X\subseteq Y\subseteq V(G)$,
it holds $|\Max(Y)|\leq |Y|$,   oracle 
$\mathrm{L}_1(X;Y)$ runs  in  $O(|Y|^2\min\{k\!+\!1,n\}m)$ 
time and $O(n^2)$ space, 
and $\mathrm{L}_2(Y)$ runs in $O(|Y|^3\min\{k\!+\!1,n\}m)$  
time and $O(n^2)$ space. 
\item[{\rm (ii)}] 
The family $\MC=\MC_{k,\mathrm{vertex}}^{\mathrm{in}}$ is transitive.  
For any non-empty subsets $X\subseteq Y\subseteq V(G)$,
it holds $|\Max(Y)|\leq {{|Y|}\choose{k}}$,   
oracle  $\mathrm{L}_1(X;Y)$ runs in  $O(|Y|^2\min\{k\!+\!1,n^{1/2}\}m)$ 
time and $O(n^2)$ space, and 
oracle $\mathrm{L}_2(Y)$ runs in 
$O(|Y|^{k+2}\min\{k\!+\!1,n^{1/2}\}m)$
time and $O(|Y|^{k}n)$ space.
\end{enumerate}
\end{lem}
\begin{proof} 
Let  $(M,w,\gamma,k,\Lambda)$  be  a  system such that 
 a mixed graph $M:=G$, $\Lambda:=V(G)$, a weight function $w$ and
  a coefficient function
   $\gamma=(\alpha,\overline{\alpha},\alpha^+,\alpha^-,\beta)$
  such that 
  $\alpha(e):=1$ and $\overline{\alpha}(e):=\alpha^+(e):=\alpha^-(e):=0$
  for each edge $e\in E(G)$, 
  and 
  $\beta(a):=0$ for each element $a\in V(G)\cup E(G)$,
  where we see that   $\gamma$ is monotone and 
  the family $\MC(M,w,\gamma,k,\Lambda)$ is transitive
  by Lemmas~\ref{lem:transitive} and \ref{lem:monotone}. 
 
 (i) We set  weight $w$ so that $w(e):=1$ for each edge $e\in E(G)$ 
 and $w(v):=k$ for each vertex $v\in V(G)$. 
We claim that 
   $\MC_{k,\mathrm{edge}}^{\mathrm{in}}$ is 
equal to  $\MC(M,w,\gamma,k,\Lambda)$,
 where the latter is the family of non-empty subsets  
  $X\subseteq \Lambda$ with 
 $\omega_X(V(X))\geq k$ such that
  $|V(X)|=1$ or $\mu(u,v;X)\geq k$ for each pair of vertices $u,v\in V(X)$.  
 Note that $\omega_X(V(X))\geq w(V(X))=k|X|$ for any non-empty set
 $X\subseteq \Lambda$.
 Then every set $X\subseteq V(G)$ with $|X|=1$ belongs
 to both 
 $\MC_{k,\mathrm{edge}}^{\mathrm{in}}$ and $\MC(M,w,\gamma,k,\Lambda)$.
Let $X$ be a subset of $V(G)$ with  $|V(X)|=|X|\geq 2$. 
By definition of coefficient function $\gamma$ and weight $w$ in $G$,
we see that 
 $\mu(u,v;X)=\lambda(u,v;G[X])$ holds 
for any two vertices $u,v\in V(X)$.  
Hence $G[X]$ is  a $k$-edge-connected graph if and only if 
$\mu(u,v;X)=\lambda(u,v;G[X])\geq k$ 
for any two vertices $u,v\in V(X)$. 
 This means that 
 $\MC_{k,\mathrm{edge}}^{\mathrm{in}}=\MC(M,w,\gamma,k,\Lambda)$,
 proving the claim. 

Whether $\mu(s,t;X)\geq k$ 
 (i.e.,   $\lambda(s,t;G[X]),\lambda(t,s;G[X])\geq k$) or not 
 for a subset $X\subseteq V(G)$
 can be tested in $O(\min\{k,n\}m)$  time \cite{AMO89,AMO93}.
By Lemma~\ref{lem:maximal}(iv),   
 $\mathrm{L}_1(X;Y)$ runs  in  $O(|Y|^2\min\{k\!+\!1,n\}m)$ 
  time and $O(n^2)$ space. 
 The family $\mathcal{K}$ of $k$-cores  $Z\subseteq Y$ is
  $\{\{v\}\mid v\in Y\}$. 
By Lemma~\ref{lem:core},
$|\Max(Y)|\leq |\mathcal{K}|\leq |Y|$ 
and  $\mathrm{L}_2(Y)$ runs in  $O(|Y|^3\min\{k\!+\!1,n\}m)$ time 
and $O(n^2)$ space.

(ii)   We set  weight $w$ so that $w(e):=1$  
for each edge $e\in E(G)$ and
 $w(v):=1$  for each vertex $v\in V(G)$. 
We claim that   $\MC_{k,\mathrm{vertex}}^{\mathrm{in}}$ is equal to  
$\MC(M,w,\gamma,k,\Lambda)$. 
Note that $\omega_X(V(X))= w(V(X))=|X|$
 for any non-empty set  $X\subseteq \Lambda$.
By definition of coefficient function $\gamma$ and weight $w$ in $G$,
we see that 
 $\mu(u,v;X)=\kappa(u,v;G[X])$ holds 
for any two vertices $u,v\in X$.  
In particular, if $|X|\leq k$ then 
$\mu(u,v;X)=\kappa(u,v;G[X])<k$. 
Let $X$ be a subset of $V(G)$ with  $|V(X)|=|X|\leq k$. 
Then $G[X]$ is not a $k$-vertex-connected graph and
$X$ is not $k$-connected in the system $(M,w,\gamma,k,\Lambda)$. 
Let $X$ be a subset of $V(G)$ with  $|V(X)|=|X|\geq k+1$. 
Then $G[X]$ is a $k$-vertex-connected graph 
if and only if 
$\mu(u,v;X)=\kappa(u,v;G[X])\geq k$ 
for any two vertices $u,v\in V(X)$. 
 This means that 
 $\MC_{k,\mathrm{vertex}}^{\mathrm{in}}=\MC(M,w,\gamma,k,\Lambda)$,
 proving the claim. 

 Whether $\mu(s,t;X)\geq k$  
  (i.e.,   $\kappa(s,t;G[X]),\kappa(t,s;G[X])\geq k$) or not 
 for a subset $X\subseteq V(G)$
  can be tested in $O(\min\{k,n^{1/2}\}m)$  time and $O(n+m)$ space \cite{AMO89,AMO93}.
By Lemma~\ref{lem:maximal}(iv),   
 $\mathrm{L}_1(X;Y)$ runs in  $O(|Y|^2\min\{k\!+\!1,n^{1/2}\}m)$  time
 and $O(n^2)$ space. 
 The family $\mathcal{K}$ of $k$-cores $Z\subseteq Y$ is ${{Y}\choose{k}}$. 
By Lemma~\ref{lem:core}, 
$|\Max(Y)|\leq |\mathcal{K}|\leq  {{|Y|}\choose{k}}$ 
and   $\mathrm{L}_2(Y)$ runs  in 
 $O(|Y|^{k+2}\min\{k\!+\!1,n^{1/2}\}m)$ time and $O(|Y|^{k}n)$ space.
\end{proof} 

Again, using \thmref{main} and \lemref{construct_oracle_inner}, 
we have the following theorem
on the time delay and the space complexity
of enumeration of connectors such that
the induced subgraphs are $k$-edge-connected or $k$-vertex-connected.

\begin{thm}  \label{thm:inner}
  Let $(G,I,\sigma)$ be an instance on a mixed graph $G$ and $k\ge0$ be an integer,
  where   $n=|V|$, $m=|E|$, and $q=|I|$. 
\begin{enumerate}
\item[{\rm (i)}] 
 All connectors that induce   $k$-edge-connected subgraphs can be enumerated 
  in $O(\min\{k\!+\!1,n\}q^2n^3m)$  delay and $O(qn+n^3)$ space. 
\item[{\rm (ii)}] 
  All connectors that induce  $k$-vertex-connected subgraphs can be enumerated  
  in $O(\min\{k\!+\!1,n^{1/2}\}q^2n^{k+2}m)$  delay and 
   $O(qn+n^{k+2})$ space. 
\end{enumerate}
\end{thm}
\begin{proof}
In both (i) and (ii), $\theta_{\rho,\mathrm{t}}$ and 
$\theta_{\rho, \mathrm{s}}$ can be regarded as $O(1)$
 since the volume function is not used anywhere in this context.

  (i) By \lemref{construct_oracle_inner}(i), 
  we have $\theta_{1,\mathrm{t}}=O(\min\{k\!+\!1,n\}n^2m)$,
  $\theta_{2,\mathrm{t}}=O(\min\{k\!+\!1,n\}n^3m)$,
  and $\theta_{1,\mathrm{s}}=\theta_{2,\mathrm{s}}=O(n^2)$,
  and we can set $\delta(Y)=n$ for any $Y\subseteq V$.
  By \thmref{main}, we have the time delay
  $O\big(q\theta_{2,\mathrm{t}} 
  + (q(n+\theta_{1,\mathrm{t}})+\theta_{\rho,\mathrm{t}})q
  \delta(V)\big)
  =O(\min\{k\!+\!1,n\}q^2n^3m)$
  and the space complexity
  $O\big((q+n+\theta_{1,\mathrm{s}}+\theta_{2,\mathrm{s}}
  + \theta_{\rho,\mathrm{s}}) n\big)
  =O(qn+n^3)$. 
  
  (ii) By \lemref{construct_oracle_inner}(ii), 
  we have $\theta_{1,\mathrm{t}}=O(\min\{k\!+\!1,n^{1/2}\}n^2m)$,
  $\theta_{2,\mathrm{t}}=O(\min\{k\!+\!1,n^{1/2}\}n^{k+2}m)$,
  $\theta_{1,\mathrm{s}}=O(n^2)$, and
   $\theta_{2,\mathrm{s}}=O(n^{k+1})$,
  and we can set $\delta(Y)=n^k$ for any $Y\subseteq V$.
  By \thmref{main}, we have the time delay 
  $O\big(q\theta_{2,\mathrm{t}} 
  + (q(n+\theta_{1,\mathrm{t}})+\theta_{\rho,\mathrm{t}})q
  \delta(V)\big)
  =O(\min\{k\!+\!1,n^{1/2}\}q^2n^{k+2}m)$
  and the space complexity 
  $O\big((q+n+\theta_{1,\mathrm{s}}+\theta_{2,\mathrm{s}}
  + \theta_{\rho,\mathrm{s}}) n\big)
  =O(qn+n^{k+2})$. 
\end{proof}

\section{Enumerating Connected Subgraphs}
\label{sec:app.subset}

As we observed in \secref{enum},
we can enumerate all components in a given transitive system
$(V,\MC)$.
This approach can be applied to enumeration
of vertex subsets that induce subgraphs under
various connectivity conditions. 

\subsection{Connected Induced Subgraphs (CISs)}

For an undirected graph $G$,
there are some studies on enumeration of CISs. 
In the seminal paper on reverse search~\cite{AF.1996}, 
Avis and Fukuda showed that
all CISs are enumerable in output-polynomial time and in polynomial space.
Their algorithm is immediately turned into a polynomial-delay algorithm
whose time complexity is $O(n)$, where $n=|V(G)|$.
Uno~\cite{U.2015} showed that all CISs are enumerable
in $O(1)$ time for each solution,
using the analysis technique called {\em Push Out Amortization}. 
Alokshiya et al.~\cite{ASA.2019} proposed a new linear delay 
algorithm and showed its empirical efficiency 
by experimental comparison with other algorithms. 
 
The above mentioned algorithms are specialized
to the task of enumerating all CISs. 
Our algorithm is so general that
it is applicable to the task by taking
the transitive system $(V(G),\MC_G)$.
Recall that, for $(V(G),\MC_G)$,
we can implement the oracles L$_1$ and L$_2$
so that $\theta_{i,\mathrm{t}}=O(n+m)$, $i=1,2$, 
$\theta_{i,\mathrm{s}}=O(n+m)$, $i=1,2$, 
and $\theta_{\rho,\mathrm{t}}=\theta_{\rho,\mathrm{s}}=O(1)$. 
 \corref{comp} implies that 
  all components in $(V(G),\MC_G)$ can be enumerated 
in $O(n^3(n+m))$ delay and $O(n(n+m))$ space. 

\subsection{$k$-Edge- and $k$-Vertex-Connected Induced Subgraphs}

For a mixed graph $G$, a subgraph $G'$ with $V(G')\subseteq V(G)$ and
 $E(G')\subseteq E(G)$ is {\em spanning} if $V(G')=V(G)$. 
There is some literature on enumeration of spanning subgraphs
 that are $k$-edge- or $k$-vertex-connected.
Khachiyan et al.~\cite{KBBEGM.2006} showed that, when
$G$ is undirected,  
all minimal 2-vertex-connected spanning subgraphs
are enumerable in incremental polynomial time. 
Boros et al.~\cite{BBEGMR.2007} extended the result
so that all minimal $k$-edge-connected (resp., $k$-vertex-connected)
 spanning subgraphs
can be enumerated in incremental polynomial time for any $k$
(resp., a constant $k$). 
Nutov~\cite{N.2009} showed that,
whether $G$ is undirected or directed,
minimal undirected Steiner networks,
and minimal $k$-vertex-connected
and $k$-outconnected spanning subgraphs
are enumerable in incremental polynomial time. 
%
Recently,
Yamanaka et al.~\cite{YMN.2019}
proposed a reverse search algorithm
that enumerates all $k$-edge-connected spanning subgraphs
of an undirected graph in polynomial delay for any $k$.

By \corref{comp} and \lemref{construct_oracle_inner},
we can enumerate all vertex subsets
that induce $k$-edge- and $k$-vertex-connected
subgraphs in a given  mixed graph $G$
since they constitute components of the transitive systems
$(V(G),\MC_{k,\mathrm{edge}}^{\mathrm{in}})$ and
 $(V(G),\MC_{k,\mathrm{vertex}}^{\mathrm{in}})$, respectively.

\begin{thm}
  \label{thm:inner_subgraph}
  Let $G$ be  a mixed graph and $k\ge0$ be an integer,
  where $n=|V(G)|$ and $m=|E(G)|$. 
\begin{enumerate}
\item[{\rm (i)}] 
All vertex subsets that induce $k$-edge-connected subgraphs can be enumerated 
  in $O(\min\{k\!+\!1,n\}n^5m)$  delay and   $O(n^3)$ space. 
\item[{\rm (ii)}] 
All vertex subsets that induce $k$-vertex-connected subgraphs can be enumerated  
  in $O(\min\{k\!+\!1,n^{1/2}\}n^{k+4}m)$   delay and  $O(n^{k+2})$ space. 
\end{enumerate}
\end{thm}

\subsection{Subgraphs Induced by Edges}

Let  $G$ be a mixed graph.
For an edge subset $F\subseteq E(G)$,
let $G[F]$ denote the subgraph $H$ induced by $F$;
i.e., $V(H)=V(F)$ and $E(H)=F$.  
Define  $\mathcal{E}_{k,\mathrm{edge}}^{\mathrm{in}}$ to be the family
(resp.,  $\mathcal{E}_{k,\mathrm{vertex}}^{\mathrm{in}}$) 
of subsets $F\in E(G)$ such that 
the induced graph $G[F]$ is $k$-edge-connected (resp., $k$-vertex-connected).
Analogously with \lemref{construct_oracle_inner}, 
we obtain the next result.
 
\begin{lem}   \label{lem:construct_oracle_edge}   
Let $G$ be a mixed graph with $n$ vertices and
$m$ edges and $k\geq 0$ be an integer.
 Then: 
\begin{enumerate}
\item[{\rm (i)}] 
The family $\MC=\mathcal{E}_{k,\mathrm{edge}}^{\mathrm{in}}$ is transitive.  
For any non-empty subsets $X\subseteq Y\subseteq E(G)$,
it holds $|\Max(Y)|\leq |Y|$,   oracle 
$\mathrm{L}_1(X;Y)$ for a subset $X\subseteq Y$
runs  in  $O(|Y|^2\min\{k+1,n\}m)$  
time and $O(n^2)$ space, 
and $\mathrm{L}_2(Y)$ runs in $O(|Y|^3\min\{k+1,n\}m)$   
time and $O(n^2)$ space.  
\item[{\rm (ii)}] 
The family $\MC=\mathcal{E}_{k,\mathrm{vertex}}^{\mathrm{in}}$ is transitive.  
For any non-empty subsets $X\subseteq Y\subseteq E(G)$,
it holds $|\Max(Y)|\leq {{|Y|}\choose{k}}$,   
oracle  $\mathrm{L}_1(X;Y)$ runs in  $O(|Y|^2\min\{k+1,n^{1/2}\}m)$  
time and $O(n^2)$ space, and 
oracle $\mathrm{L}_2(Y)$ runs in 
$O(|Y|^{k+2}\min\{k+1,n^{1/2}\}m)$ 
time and $O(|Y|^{k}n)$ space.
\end{enumerate}
\end{lem}
\begin{proof} 
Let  $(M,w,\gamma,k,\Lambda)$  be  a  system such that 
 a mixed graph $M:=G$, $\Lambda:=E(G)$, a weight function $w$ and
  a coefficient function
   $\gamma=(\alpha,\overline{\alpha},\alpha^+,\alpha^-,\beta)$
  such that 
  $\alpha(e):=\overline{\alpha}(e):=\alpha^+(e):=\alpha^-(e):=0$
  for each edge $e\in E(G)$, 
  and 
  $\beta(a):=0$ for each element $a\in V(G)\cup E(G)$,
  where we see that   $\gamma$ is monotone and 
  the family $\MC(M,w,\gamma,k,\Lambda)$ is transitive
  by Lemmas~\ref{lem:transitive} and \ref{lem:monotone}. 
 
 (i) We set  weight $w$ so that $w(e):=1$ for each edge $e\in E(G)$ 
 and $w(v):=k$ for each vertex $v\in V(G)$. 
We claim that 
   $\mathcal{E}_{k,\mathrm{edge}}^{\mathrm{in}}$ is 
equal to  $\MC(M,w,\gamma,k,\Lambda)$,
 where the latter is the family of non-empty subsets  
  $X\subseteq \Lambda$ with 
 $\omega_X(V(X))\geq k$ such that
  $|V(X)|=1$ or $\mu(u,v;X)\geq k$ for each pair of vertices $u,v\in V(X)$.  
 Note that $|V(X)|\geq 2$ and
  $\omega_X(V(X))= w(V(X))=k|V(X)|$ for any non-empty set
 $X\subseteq \Lambda=E(G)$.
 %
Let $X$ be a non-empty subset of $E(G)$, 
where  $|V(X)|\geq 2$ and   $\omega_X(V(X))\geq k$. 
By definition of coefficient function $\gamma$ and weight $w$ in $G$,
we see that 
 $\mu(u,v;X)=\lambda(u,v;G[X])$ holds 
for any two vertices $u,v\in V(X)$.  
Hence $G[X]$ is  a $k$-edge-connected graph if and only if 
$\mu(u,v;X)=\lambda(u,v;G[X])\geq k$ 
for any two vertices $u,v\in V(X)$. 
 This means that 
 $\mathcal{E}_{k,\mathrm{edge}}^{\mathrm{in}}
 =\MC(M,w,\gamma,k,\Lambda)$,
 proving the claim. 

Whether $\mu(s,t;X)\geq k$ 
 (i.e.,   $\lambda(s,t;G[X]),\lambda(t,s;G[X])\geq k$) or not 
 for a subset $X\subseteq E(G)$
  can be tested in $O(\min\{k,n\}m)$  time \cite{AMO89,AMO93}.
By Lemma~\ref{lem:maximal}(iv),   
 $\mathrm{L}_1(X;Y)$ runs  in  $O(|Y|^2\min\{k+1,n\}m)$ 
  time and $O(n^2)$ space. 
 The family $\mathcal{K}$ of $k$-cores  $Z\subseteq Y$ is
  $\{\{v\}\mid v\in Y\}$. 
By Lemma~\ref{lem:core},
$|\Max(Y)|\leq |\mathcal{K}|\leq |Y|$ 
and  $\mathrm{L}_2(Y)$ runs in  $O(|Y|^3\min\{k+1,n\}m)$ time 
and $O(n^2)$ space.

(ii)  We set  weight $w$ so that $w(e):=1$  
for each edge $e\in E(G)$ and
 $w(v):=1$  for each vertex $v\in V(G)$.
We claim that  $\mathcal{E}_{k,\mathrm{vertex}}^{\mathrm{in}}$ is equal to  
$\MC(M,w,\gamma,k,\Lambda)$. 
Note that  $|V(X)|\geq 2$ and
$\omega_X(V(X))= w(V(X))=|X|$
 for any non-empty set  $X\subseteq \Lambda$.
By definition of coefficient function $\gamma$ and weight $w$ in $G$,
we see that 
 $\mu(u,v;X)=\kappa(u,v;G[X])$ holds 
for any two vertices $u,v\in V(X)$.  
In particular, if $|V(X)|\leq k$ then 
$\mu(u,v;X)=\kappa(u,v;G[X])<k$. 
Let $X$ be a subset of $E(G)$ with  $|V(X)| \leq k$. 
Then $G[X]$ is not a $k$-vertex-connected graph and
$X$ is not $k$-connected in the system $(M,w,\gamma,k,\Lambda)$. 
Let $X$ be a subset of $E(G)$ with  $|V(X)|\geq k+1$. 
Then $G[X]$ is a $k$-vertex-connected graph 
if and only if 
$\mu(u,v;X)=\kappa(u,v;G[X])\geq k$ 
for any two vertices $u,v\in V(X)$. 
 This means that 
 $\mathcal{E}_{k,\mathrm{vertex}}^{\mathrm{in}}
 =\MC(M,w,\gamma,k,\Lambda)$,
 proving the claim.  

 Whether $\mu(s,t;X)\geq k$  
  (i.e.,   $\kappa(s,t;G[X]),\kappa(t,s;G[X])\geq k$) or not 
 for a subset $X\subseteq E(G)$
 can be tested in $O(\min\{k,n^{1/2}\}m)$  time and $O(n+m)$ space \cite{AMO89,AMO93}.
By Lemma~\ref{lem:maximal}(iv),   
 $\mathrm{L}_1(X;Y)$ runs in  $O(|Y|^2\min\{k+1,n^{1/2}\}m)$  time
 and $O(n^2)$ space. 
 The family $\mathcal{K}$ of $k$-cores $Z\subseteq Y$ is ${{Y}\choose{k}}$. 
By Lemma~\ref{lem:core}, 
$|\Max(Y)|\leq |\mathcal{K}|\leq  {{|Y|}\choose{k}}$ 
and   $\mathrm{L}_2(Y)$ runs  in 
 $O(|Y|^{k+2}\min\{k+1,n^{1/2}\}m)$ time and $O(|Y|^{k}n)$ space.
\end{proof} 

By \corref{comp} and \lemref{construct_oracle_edge},
we can enumerate all edge subsets
that induce $k$-edge- and $k$-vertex-connected
subgraphs in a given mixed graph $G$
since they constitute components of the transitive systems
$(E(G),\mathcal{E}_{k,\mathrm{edge}}^{\mathrm{in}})$ and
 $(E(G),\mathcal{E}_{k,\mathrm{vertex}}^{\mathrm{in}})$, respectively.

\begin{thm}
  \label{thm:edge_induced_subgraph}
  Let $G$ be a mixed graph and $k\ge1$ be an integer,
  where $n=|V(G)|$ and $m=|E(G)|$. 
\begin{enumerate}
\item[{\rm (i)}] 
All edge subsets that induce $k$-edge-connected subgraphs 
can be enumerated 
  in $O(\min\{k+1,n\}m^6)$  delay and   $O(mn^2)$ space. 
\item[{\rm (ii)}] 
All edge subsets that induce $k$-vertex-connected subgraphs can be enumerated  
  in $O(\min\{k+1,n^{1/2}\}m^{k+5})$   delay and  $O(m^{k+1}n)$ space. 
\end{enumerate}
\end{thm}

Define a volume function $\rho:V(G)\cup E(G)\to \mathbb{R}$
so that $\rho(X):=|V(X)|- |V(G)|+1$ for each subset $X\subseteq V(G)\cup E(G)$.
For a subset $X\subseteq E(G)$, 
the graph $G[X]$ is a spanning subgraph of $G$ if and only if
$\rho(X)> 0$. 
We see that 
$\theta_{\rho,\mathrm{t}}=\theta_{\rho,\mathrm{s}}=O(n)$. 
 
Similarly to \thmref{edge_induced_subgraph},
we can enumerate all edge subsets
that induce $k$-edge- and $k$-vertex-connected
spanning subgraphs in a given  mixed graph $G$
since they constitute $\rho$-positive components of the transitive systems
$(E(G),\mathcal{E}_{k,\mathrm{edge}}^{\mathrm{in}})$ and
 $(E(G),\mathcal{E}_{k,\mathrm{vertex}}^{\mathrm{in}})$, respectively.

\begin{thm}
  \label{thm:spanning_subgraph}
  Let $G$ be  a mixed graph and $k\ge1$ be an integer,
  where $n=|V(G)|$ and $m=|E(G)|$. 
\begin{enumerate}
\item[{\rm (i)}] 
All edge subsets that induce $k$-edge-connected spanning subgraphs 
can be enumerated 
  in $O(\min\{k+1,n\}m^6)$  delay and   $O(mn^2)$ space. 
\item[{\rm (ii)}] 
All edge subsets that induce $k$-vertex-connected spanning subgraphs can be enumerated  
  in $O(\min\{k+1,n^{1/2}\}m^{k+5})$   delay and  $O(m^{k+1}n)$ space. 
\end{enumerate}
\end{thm}

\begin{table}[b!]
  \centering
  \caption{Complexity of enumerating connectors $X$ that satisfy 
  several connectivity requirements}
  \label{tab:result}
  \begin{tabular}{ccll}
    \hline
    \multicolumn{1}{l}{\bf Theorem} &
    \multicolumn{1}{c}{\bf  Requirement} &
    \multicolumn{1}{c}{\bf Delay} &
    \multicolumn{1}{c}{\bf Space} \\
    \hline
    \ref{thm:connect} & $G[X]$ is connected & $O(q^2(n+m)n)$ & $O((q+n+m)n)$\\
    \ref{thm:entire}(i) & $X$ is $k$-edge-connected & $O(q^2n^3)$ & $O(qn+n^3)$ \\
    %
    %
    \ref{thm:entire}(ii) & $X$ is $k$-vertex-connected & 
    $O(q^2n^{k+2})$ & $O(qn+n^{k+2})$\\
    \ref{thm:inner}(i) & $G[X]$ is $k$-edge-connected & 
    $O(\min\{k\!+\!1,n\}q^2n^3m)$ & $O(qn+n^3)$ \\
    \ref{thm:inner}(ii) & $G[X]$ is $k$-vertex-connected & 
    $O(\min\{k\!+\!1,n^{1/2}\}q^2n^{k+2}m)$ & $O(qn+n^{k+2})$ \\
    \hline
  \end{tabular}
\end{table}

\begin{table}[b!]
  \centering
  \caption{Complexity of enumerating vertex subsets $X$ or
  edge subsets $F$ that 
  satisfy several connectivity requirements}
  \label{tab:result2}
  \begin{tabular}{c c ll}
    \hline
    \multicolumn{1}{l}{\bf Theorem} &
    \multicolumn{1}{c}{\bf  Requirement} &
    \multicolumn{1}{c}{\bf Delay} &
    \multicolumn{1}{c}{\bf Space} \\
    \hline
    \ref{thm:inner_subgraph}(i) & $X\subseteq V(G)$, $G[X]$ is $k$-edge-connected 
    & $O(\min\{k\!+\!1,n\}n^5m)$ & $O(n^3)$\\
    \ref{thm:inner_subgraph}(ii) &$X\subseteq V(G)$, $G[X]$ is $k$-vertex-connected 
    & $O(\min\{k\!+\!1,n^{1/2}\}n^{k+4}m)$ & $O(n^{k+2})$ \\
    \ref{thm:edge_induced_subgraph}(i) &$F\subseteq E(G)$, 
    $G[F]$ is $k$-edge-connected & $O(\min\{k\!+\!1,n\}m^6)$ & $O(mn^2)$ \\
    \ref{thm:edge_induced_subgraph}(ii) &$F\subseteq E(G)$, 
    $G[F]$ is $k$-vertex-connected  & $O(\min\{k\!+\!1,n^{1/2}\}m^{k+4})$ & $O(m^{k+1}n)$ \\
    \ref{thm:spanning_subgraph}(i) &$F\subseteq E(G)$, 
    $(V(G),F)$ is $k$-edge-connected & $O(\min\{k\!+\!1,n\}m^6)$ & $O(mn^2)$ \\
    \ref{thm:spanning_subgraph}(ii) &$F\subseteq E(G)$, 
    $(V(G),F)$ is $k$-vertex-connected  & $O(\min\{k\!+\!1,n^{1/2}\}m^{k+4})$ & $O(m^{k+1}n)$ \\
    \hline
  \end{tabular}
\end{table}

\section{Concluding Remarks}
\label{sec:conc}
The main contribution of the paper is \thmref{main}.
To prove the theorem,
we have presented a family-tree based enumeration algorithm
that achieves the required complexity in Sections~\ref{sec:defn} and \ref{sec:trav}.

Once we define a transitive system $(V,\MC)$
such that $\delta(X)\le\delta(Y)$ holds for any $X\subseteq Y\subseteq V$
and design two oracles L$_1$ and L$_2$ for it,
we can enumerate all solutions
in an instance $(V,\MC,I,\sigma)$ for arbitrary $I$ and $\sigma$
 in the stated computational complexity.  
In particular, if the time (resp., space) complexity of the two oracles
is polynomially bounded, the algorithm achieves polynomial-delay
(resp., polynomial space). 

We presented some application results in
Sections~\ref{sec:app.conn} and \ref{sec:app.subset}.
In \secref{app.conn},
we obtained the first polynomial-delay algorithm
for the connector enumeration problem,
even when a stronger connectivity condition is imposed
on a connector (i.e., $k$-edge-connectivity for any $k$ and $k$-vertex-connectivity 
for a fixed $k$). 
In \secref{app.subset},
we showed that all vertex subsets that induce $k$-edge-connected
 (resp., $k$-vertex-connected)
subgraphs are enumerable in polynomial delay for any $k$ (resp., a fixed $k$). 
We summarize the computational complexity in Table 1 and 2. 
We could improve complexity bounds for respective cases,
which are left for future work.

Our next issue is to show the effectiveness of
the family-tree based algorithm 
by solving real instances of enumeration problems
concerning a transitive system. 
The connector enumeration problem has applications in biology,
as mentioned in \secref{app.conn}, and
we are to pursuit further applications in such fields as chemistry. 
We have already developed an implementation of the algorithm
for this problem
and observed its efficiency in comparison with previous algorithms,
COOMA and COPINE~\cite{HN.2019}.


\end{document}